\documentclass[reqno,a4paper,9pt]{amsart}
\usepackage{marginnote}
\usepackage[a4paper, centering]{geometry}
\usepackage[utf8]{inputenc} 
\usepackage[english]{babel}
\usepackage[T1]{fontenc}
\usepackage{lmodern}
\usepackage{caption}
\usepackage{subcaption}
\usepackage{enumerate}
\usepackage{amssymb, amsmath}
\usepackage{mathrsfs,dsfont}
\usepackage{amscd}
\usepackage{bbm}
\usepackage{amsthm}
\usepackage{ esint }
\usepackage[mathcal]{euscript}
\usepackage[active]{srcltx}
\usepackage{verbatim}
\usepackage[colorlinks,linkcolor={blue},citecolor={blue},urlcolor={black}]{hyperref}
\usepackage[]{changebar}
\usepackage{enumitem}
\usepackage{xcolor}

\usepackage{mathtools}
\usepackage{hyperref}
\usepackage{url}            
\usepackage{booktabs}       
\usepackage{amsfonts}       
\usepackage{nicefrac}       
\usepackage{microtype}      
\usepackage{lipsum}
\usepackage{fancyhdr}       
\usepackage{graphicx}       
\graphicspath{{media/}}     

\usepackage[numbers]{natbib}

\theoremstyle{plain}
\newtheorem{theorem}{Theorem}[section]
\newtheorem{corollary}[theorem]{Corollary}

\newtheorem*{definition*}{Definition}

\theoremstyle{remark}

\newtheorem{example}[theorem]{Example}
\newtheorem*{remark*}{Remark}
\newtheorem*{example*}{Example}
\newtheorem*{notation*}{Notation}

\numberwithin{equation}{section}

\def\R{{\mathbb R}}

\def\P{{\mathbb P}}

\def\B{{\rm BF}}
\def\bic{{\rm BIC}}
\def\nbic{{\rm BIC\_HES}}
\def\det{{\rm det}}

\definecolor{viola}{rgb}{0.3,0,0.7}
\definecolor{ciclamino}{rgb}{0.5,0,0.5}
\definecolor{rosso}{rgb}{0.8,0,0}

\newcommand{\beq}{\begin{equation}}
\newcommand{\eeq}{\end{equation}}
\newcommand{\bal}{\begin{aligned}}
\newcommand{\eal}{\end{aligned}}
\newcommand{\ben}{\begin{enumerate}}
\newcommand{\beni} {\begin{enumerate}[(i)]}
\newcommand{\een}{\end{enumerate}}
\newcommand{\bit}{\begin{itemize}}
\newcommand{\eit}{\end{itemize}}
\newcommand{\beqw}{\begin{equation*}}
\newcommand{\eeqw}{\end{equation*}}
\newcommand{\bex}{\begin{example}}
\newcommand{\eex}{\end{example}}
\newcommand{\bre}{\begin{example}}
\newcommand{\ere}{\end{example}}
\newcommand{\bma}{\begin{bmatrix}}
\newcommand{\ema}{\end{bmatrix}}

\makeatletter
\@namedef{subjclassname@2020}{%
  \textup{2020} Mathematics Subject Classification}
\makeatother

\title[Bicnew]{A Modified Bayesian Criterion for Model Selection in Mixed and Hierarchical Frameworks}

\author[D. D. J. Ramirez Ramirez]{Diogenes de Jesus Ramirez\textsuperscript{*}}
\address[D. D. J. Ramirez Ramirez]{Department of Mathematics,
National University of Colombia, Manizales Campus, Manizales, Colombia.}
\email{ddramirezra@unal.edu.co}
\thanks{\textsuperscript{*}Corresponding author: Name: Diogenes De Jesus; Surname: Ramirez Ramirez; email: ddramirezra@unal.edu.co}

\author[A.~Melchor Hernandez]{Anderson Melchor Hernandez}
\address[A.~Melchor Hernandez]{Department of Mathematics, University of Bologna, Piazza di Porta San Donato 5, 40126, Bologna (Italy)}
\email{anderson.melchor@unibo.it}

\author[Isabel Cristina Ramírez]{Isabel Cristina Ramírez}
\address[Isabel Cristina Ramírez]{Departamento de Matematicas,
National University of Colombia, Medellin Campus, Medellin, Colombia.}
\email{iscramirezgu@unal.edu.co}

\author[Luis Raúl Pericchi]{Luis Raúl Pericchi}
\address[Luis Raúl Pericchi]{University of Puerto Rico, Río Piedras Campus, San Juan, Porto Rico}
\email{luis.pericchi@upr.edu}

\date{\today}

\begin{document}

\begin{abstract}
In this work, we propose a modified Bayesian Information Criterion (BIC) specifically designed for mixture models and hierarchical structures. This criterion incorporates the determinant of the Hessian matrix of the log-likelihood function, thereby refining the classical Bayes Factor by accounting for the curvature of the likelihood surface. Such geometric information introduces a more nuanced penalization for model complexity. The proposed approach improves model selection, particularly under small-sample conditions or in the presence of noise variables. Through theoretical derivations and extensive simulation studies—including both linear and linear mixed models—we show that our criterion consistently outperforms traditional methods such as BIC, Akaike Information Criterion (AIC), and related variants. The results suggest that integrating curvature-based information from the likelihood landscape leads to more robust and accurate model discrimination in complex data environments.
\end{abstract}

\maketitle

\tableofcontents

\section{Introduction}
The increasing need for information storage and the challenges that come with its analysis imply new forecasting challenges for the various agents of knowledge and their specific areas of expertise. Therefore, biostatistics, epidemiology, economics, psychology, sociology, and many other branches of human knowledge face daily challenges within their organizations, not only in terms of classified information but also in order to implement and develop research based on the large volume of stored data to find patterns of behavior. Typically, the stored information is of great interest to the various actors in society who collect and store it, as mathematical models can be constructed from a set of variables that closely emulate the observed reality. These models can be evaluated by implementing a broad spectrum of hypothesis tests that largely allow for classification, individualization, and selection of a model whose mathematical context allows for a coherent link between the researcher and the simulated reality. However, any statistical test loses not only its meaning but also its significance to the extent that the exact manner in which nature generates a set of data and the existing functionality of interaction between the characteristics or variables that are part of the model is unknown. Continuing with what has been described, the need to have and implement statistical models is becoming increasingly useful once linked with modern information and communication technologies. These models allow not only academic but also business agents to make decisions in an increasingly controlled environment of uncertainty, and that is when science and academia become relevant. In light of the above, statistics and its way of using a set of data through the implementation of well-known methods and techniques not only allow for inference through the use of a specific data sample but also prompt thinking about the behavior of a specific population, sample, or data set with a measurable risk of error in terms of probability. On the other hand, the implementation of statistical models in everyday life is, according to \cite{wallas}, an approximation to reality based on a mathematically unknown truth. For the authors, there are no true models that can manifest a reality that is not only perfect but also complete. The specification of a correct model is important for making consistent and efficient inferences. Unfortunately, it is not always easy to know if one has a correct model, and thus one needs to compare models and decide which one is best among them, i.e., the model that is most supported by the data. If the models are nested, the task is particularly simple, so model selection is equivalent to a parameter test \cite{pawitan2001}. In the last years, linear mixed models have been extensively used for the analysis of longitudinal data in medicine, psychology, and other fields \cite{gelman2021,l1990, vehtari2017}. Such a class of statistical models include fixed and random effects which endow these models with a hierarchical structure. These models are typically used when there is no independence in the data. For instance, let us think about students  which could be sampled from within classrooms. The core of mixed models is that they incorporate fixed and random effects. A fixed effect is a parameter that does not vary. For example, we may assume there is some true regression line in a population. In contrast, random effects are parameters that are themselves random variables. In the last years, a prominent interest has been given to Bayesian linear mixed models, and specifically on Bayesian factor hypothesis tests \cite{gabry2020,rouder2012,watanabe2010}. However, the availability of a huge literature about the Bayesian procedure, it remains a considerable ambiguity on the most appropriate Bayes factor hypothesis that must then be applied to quantify the number of variables needed to support an experiment effect, more specifically, when considering different random terms in the model.\\
In this work, we introduce a perturbed Bayesian rule that can helps us to decide how to exclude variables from a data which can be considered as usefulness for a determined model. In the scientific community,  several times scientists are trained in  the classical null hypothesis significance testing, in short (NHST). Let us for instance consider the case where one needs to test a specific test hypothesis (e.g., the effect of a vaccine over an ancient community). In the classical approach, NSHT works by considering a null hypothesis (e.g., the vaccine has already showed an effect) and thus one needs to verify such a requirement throughout a sample data. Such a statistic test is usually compared with a sample of distributions that would arise if the null hypothesis were correct. However such a approach have some limitations, say, when we can reject the null hypothesis, but we cannot simply accept the null hypothesis. Henceforth, To overcome such a limitation, we can fruitfully exploit the so-called Bayes factor indicator, in short, $\bic$ \cite{faulk2018}. Hence, the use of this criteria needs the specification of a null model ( say not effect of a vaccine), and an alternative model. In the frequentistic world, a useful approach for the specification of the alternative mixed effects model of the full data is to ask for a model that includes all fixed and random effects required by the experiment (i.e., to consider the maximal model)\cite{van2021bayes}.  Even though there are several null models that the maximal model can
be compared to, it is quite appropriate the specification of the null because it defines the question we ask about the effects. For instance, different decisions need to be made when testing for the effect of a manipulation in an experiment, and how
to quantify effects, and how to set prior distributions. After fitting a Bayesian model, one usually wants to measure its predictive accuracy.  In the recent years, a trend of cross validation and information technique's have  been used \cite{akaike1998,piironen2017,spiegelhalter1996} to fit models. In the case of a cross-validation approach, one requires re-fitting the model with different training sets. The Watanabe-Akaike information criterion (WAIC) is an another statistical approach which can be seen as an improvement of the deviance information criterion (DIC) for Bayesian models.  WAIC is a Bayesian method and it uses the entire posterior distribution of parameters. Furthermore, it is invariant under parametrizations \cite{gelman2014understanding,spiegelhalter2014deviance,watanabe2010}.\\
{\bf Our contribution}\\
In this work, we address Bayesian inference for mixed-effects models by proposing a modified Bayes factor hypothesis test. Specifically, we extend the usual Bayesian Information Criterion (BIC) by incorporating an additional term involving the determinant of the Hessian matrix of the log-likelihood evaluated at its maximum. This modification is motivated by the observation that, in some models, the optimization of the log-likelihood function can be particularly challenging or even intractable. The inclusion of the determinant term provides a way to account for the local curvature of the likelihood surface and, consequently, for the amount of information contained in the model. Formally, for a given model hypothesis $H$, we define the \emph{perturbed BIC} as
\begin{align}
\nbic(H) := \bic(H) + \log\left(\det(P)\right),
\end{align}
where $P$ denotes the Fisher information matrix, $\det(\cdot)$ its determinant, and $\log(\cdot)$ the natural logarithm. The inclusion of the additional $\log(\det(P))$ term is not arbitrary; apart from its mathematical derivation (see Corollary \ref{1cor} and \cite[Section 7.2.2]{held2014applied}), it captures structural aspects of hierarchical or heterogeneous data that the standard BIC neglects. To illustrate this, consider the following simple mixed-effects model:
\begin{align*}
X_{il} = \mu_i + \epsilon_{il}, \quad \text{for } i = 1, \dots, p \text{ groups and } l = 1, \dots, r \text{ observations per group,}    
\end{align*}
where $\epsilon_{il} \sim N(0, \sigma^2)$ are independent and identically distributed random variables. Computing $\log(\det(P))$ in this setting (see Section \ref{sec:examples}) yields
\begin{align*}
\log\left(\det(P)\right) = (p + 1)\log(r) + \log(p) - \log(2) - (p + 2)\log(\hat{\sigma}^2).
\end{align*}
In large-sample contexts, $\log(\det(P))$ is often neglected, as its contribution becomes asymptotically negligible. However, in small-sample settings this term can be substantial, introducing an explicit data dependence through the empirical estimator $\hat{\sigma}^2$. Notably, when a model fits the data exceedingly well (i.e., when $\hat{\sigma}^2$ is very small), the term $-(p+2)\log(\hat{\sigma}^2)$ acts as a strong positive penalty. This behavior discourages overfitting by penalizing models that explain random noise rather than genuine structure in the data. Hence, it becomes clear that the additional $\log(\det(P))$ term introduces an adaptive penalty that depends on the hierarchical structure of the data. Indeed, $\det(P)$ reflects the overall amount of information in the model parameters, which in this simple setting depends on the number of groups $p$, the number of observations per group $r$, and the residual variance $\hat{\sigma}^2$. The resulting criterion therefore penalizes models not only for their dimensionality, as the classical BIC does, but also for their internal informational structure \cite{amari2016information}. This yields a more flexible adjustment, particularly relevant in unbalanced or small-sample mixed-effects models, where the Fisher information is not homogeneous across parameters. Finally, in analogy with the {\rm LOO} method discussed in \cite{vehtari2017}, we consider a predictive approximation of our Bayesian criterion through the computed log pointwise predictive density, denoted by $\widehat{\mathrm{lpd}}$.\\
{\bf The paper is organized as follows.}\\
In Section \ref{sec:backgorund}, we recall key concepts related to the Bayes factor. Section \ref{sec:perturbation} introduces our proposed criterion and presents Corollary \ref{1cor} and Theorem \ref{betterthm}, which provide the theoretical motivation for our modified Bayes criterion. Furthermore, \autoref{thm:consistency} guarantees its consistency.

Section \ref{sec:examples} demonstrates the practical effectiveness of the criterion by analyzing a model with random effects, where we explicitly compute the coefficient $\log(\det(P))$. In Subsection \ref{comparison}, we compare our approach with the classical BIC and highlight its advantages, particularly in small-sample scenarios.

Section \ref{simulations} presents simulation studies involving linear, mixed, binomial, and Poisson models to evaluate the performance of our criterion against other established information criteria. Subsection \ref{sec:linearcase} focuses on the linear case, while Subsection \ref{sec:mixedcase} addresses mixed models. Section \ref{sec:binomialpoisson} covers binomial and Poisson models.

In Section \ref{sec:proofs}, we provide detailed proofs of Corollary \ref{1cor}, Theorem \ref{betterthm}, and \autoref{thm:consistency}. Finally, Section \ref{sec:futureworks} outlines potential directions for future research.

\section{Background and Related Works}\label{sec:backgorund}
In the next, we recall the main feature about the Bayes factor closely following \cite{faulk2018}.
\subsection{The Bayes factor}
Bayesian inference is a method of measurement that is based on the computation of $\P(H\vert D)$, which is the posterior of a hypothesis $H$, given data $D$. Bayes' theorem casts this probability as

\begin{align}\label{b1}
    \P(H\vert D)=\frac{\P(D\vert H)\cdot \P(H)}{\P(D)}
\end{align}
where $\P(H)$ is usually referred to as prior probability to hypothesis (or sample) $H$. Using Bayes' theorem, it is trivial to see that

\begin{align}\label{b2}
    \frac{\P(H_{0}\vert D)}{\P(H_{1}\vert D)}=\frac{\P(D\vert H_{0})}{\P(D\vert H_{1})}\cdot\frac{\P(H_{0})}{\P(H_{1})}.
\end{align}
In this equation the well-known Bayes factor is

\begin{align}\label{Bayes}
   \B_{0,1}:= \frac{\P(D\vert H_{0})}{\P(D\vert H_{1})}.
\end{align}
It is important to noticing dependencies on subscripts.  Similary, we can see that 
\begin{align*}
    \B_{1,0}=\frac{1}{\B_{0,1}},
\end{align*}
and more in general, this relation for every suitable hypothesis $H_{i}$ where $i$ is an opportune index. While equation \eqref{Bayes}
may seem apparently quite simple, it is much difficult to compute with practical examples. This is because to compute the numerator or denominator one needs to introduce several parametrizations. To be more precise, the above formula \eqref{Bayes} usually is written as

\begin{align}\label{Bayes1}
    \B_{0,1}=\frac{\int_{\Theta_{0}}\P(D\vert H_{0},\theta)g_{0}(\theta)d\theta}{\int_{\Theta_{1}}\P(D\vert H_{1},\theta)g_{1}(\theta)d\theta}
\end{align}
where $\Theta_{0}$ and $\Theta_{1}$ are parameter spaces for models $H_{0}$ and $H_{1}$, respectively, and $g_{0}$ and $g_{1}$ are the prior probability density functions of the parameters of $H_{0}$ and $H_{1}$, respectively. Notice that to computing $\B_{0,1}$, one must specify the priors $g_{0}$ and $g_{1}$. This however cause a theoretical problem of estimation and thus $\B_{0,1}$ becomes to be inaccessible without such a previous information. For this reason, in the literature various  classification schemes have been introduced. In \cite{raftery1995} has used a classification where $\bic$ between $1$ and $3$ are considered weak evidence, those between $3$ to $20$ can be seen as a positive evidence, and while those in the range [$20$, $150$] constitute strong evidence, and finally beyond $150$ very strong evidence. In some sense, $\bic$ only provides an index of preference for one hypothesis rather other one when the data set does not change a lot the number of variables. When the models being treated have mismatches with the data set, sometimes one needs to evaluate their predictive behavior, and thus to estimate their error \cite{vehtari2002}. 

\section{A perturbation of Bayes factor}\label{sec:perturbation}

For a given model $H_{i}$, the $\bic$ is defined as

\begin{align*}
    \bic(H_{i}):=-2\log(L_{i}) +p_{i}\cdot \log(n),
\end{align*}
where $n$ is the number of observations, $p_{i}$ is the number of free parameters of model $H_{i}$, and $L_{i}$ is the maximum likelihood for model $H_{i}$. In \cite{wagenmakers2007} proposed that
\begin{align*}
    \B_{0,1}\approx {\rm exp \left(\bic\left(\Delta_{1,0}/2\right)\right)}
\end{align*}
where $\bic(\Delta_{1,0}):=\bic(H_{1})-\bic(H_{0})$.
This however, does not enable to consider $H_{0}$ and $H_{1}$ as the product of a same stochastic process but mathematically being different objects. Suppose that $\hat{\theta}_{1}$ and $\hat{\theta}_{0}$ are the parameter whose maximize $L_{1}$ and $L_{0}$, respectively. To overcome such a technical point, we now aim to propose a different approach to enrich the Bayes factor. 
\subsection{The proposed criterion}
Here, we propose to take the following:
 \begin{align}\label{newbic}
    \nbic(H):= \bic(H) +log\left(\det(P)\right)
     \end{align}
 where $P$ is the Fisher information matrix of size $p\times p$ which is specified by a statistical model $H$ depending on $p$ parameters. As anticipated in the introduction, the term $log\left(\det(P)\right)$ can be justified by looking at a second-order expansion of the likelihood function. Let us recall that, from an analytical point of view, the Fisher Information can be regarded as the minus Hessian matrix of the log-likelihood function. Determining whether its eigenvalues are positive or negative therefore allows us to classify its stationary points. Hence, the nature of equation (\ref{newbic}) is that it takes into account the "geometry" of the space where the parameters lie. Nevertheless, a price to pay is that this criterion would then add a further term of $p\log(n)$ where $n$ represents the sample size of the model $H$. We collect this observation in the following Corollary:

 \begin{corollary}\label{1cor}
 Let $\Theta_{i}$ be a sample set in $\R^{d}$ depending of family of parameters of size $p$ and suppose that it is invariant by rotations ( i.e., for all square matrix $P$ such that $P^{2}=P$, then $P\Theta=\Theta$). Suppose that $\tilde{\theta}_{i}\in \Theta$ is the vector of parameters which maximizes the likelihood function ${\rm log}\left( f_{i}(\cdot \vert \theta)g_{i}(\theta)\right)$ and such a function is sufficiently smooth. Then there exists a definite positve matrix $P_{i}$ of size $p\times p$ such that
 \begin{align*}
   \P(y\vert H_{i})=f_{i}(y\vert \tilde{\theta}_{i})g_{i}(\tilde{\theta}_{i})\frac{(2\pi)^{p/2}}{\det(P_{i})^{1/2}}.
 \end{align*}
 \end{corollary}
 With this Corollary, we can compare different models as follows.

\begin{theorem}\label{betterthm}
 Let us suppose that $\Theta_{i}$ is invariant by rotations, that is, for any orthogonal matrix $U$, $U(\Theta_{i})=\Theta_{i}$. Moreover, suppose that $\log(f_{i}(\vert H_{i})g_{i})$ are $C^{2}$ over all $\Theta_{i}$ and admits a unique stationary point for $i=0,1$.
 Then
 \begin{align}\label{logbif}
     \log(\B_{0,1})=\log\left(\frac{f_{0}(y\vert \tilde{\theta}_{0})g_{0}(\tilde{\theta}_{0})}{f_{1}(y\vert \tilde{\theta}_{1})g_{1}(\tilde{\theta}_{1})}\right)+\frac{1}{2}\log(2\pi)\left(\vert \tilde{\theta}_{0}\vert- \vert \tilde{\theta}_{1}\vert\right)
     +
     \frac{1}{2}\log\left( \frac{\det(P_{1})}{\det(P_{0})}\right)
 \end{align}
 where $P_{i}$ is Fisher information matrix given by $-\nabla^{2}\log \left(f_{i}(y\vert \theta)g_{i}(\theta)\right)$ at $\tilde{\theta}_{i}$, and $\vert \tilde{\theta}_{i}\vert={\rm length}(\tilde{\theta}_{i})$ for $i=0,1.$ Furthermore, if we suppose that $\vert \tilde{\theta}_{1}\vert=\vert \tilde{\theta}_{0}\vert$, then \eqref{logbif} looks like
 
 \begin{align}\label{betterbic}
     2\log(\B_{0,1})= \nbic(\Delta_{1,0}).
 \end{align}
 \end{theorem}
 As an application of \autoref{betterthm}, let us consider data $y_{1},\ldots, y_{n}$ modeled as independent given parameters $\theta$. To maintain comparability with a given dataset and to get easier interpretation of \eqref{logbif}, we need an explicit representation of $\tilde{\theta}_{0}$ and $\tilde{\theta}_{1}$, respectively.  Suppose we have a prior distribution $\P(\theta)$, thus yielding a posterior distribution $\P(\theta\vert y)$ and a posterior predictive distribution $\P(\tilde{y}\vert y)=\int \P(\tilde{y}\vert \theta)\P(\theta\vert y)dy$. Sometimes it is useful to use the  so-called {\rm log pointwise predictive density} denoted as {\rm lpd}, and which is given by

\begin{align}\label{lpd1}
   {\rm  lpd}\coloneqq\sum_{i=1}^{n}\log(\P(y_{i}\vert y))=\sum_{i=1}^{n}\log\left(\int \P(y_{i}\vert \theta)\P(\theta\vert y) d\theta\right). 
\end{align}
This is an overestimated measure for the expected  log pointwise predictive density  for a new dataset. Notice that to compute the {\rm lpd} in practice, we need draws from the posterior distribution, which we label $\theta^{j}$, $j=1,\ldots p$, $p\in \mathbb{N}$. Thus {\rm lpd} can be approximated as

\begin{align}\label{alpd}
    \widehat{{\rm lpd}}=\sum_{i=1}^{n}\log\left(\frac{1}{p}\sum_{j=1}^{p}\P(y_{i}\vert \theta^{j})\right).
\end{align}
Hence, it is possible to estimate \eqref{logbif} as

\begin{align}\label{nwlogbif}
\begin{aligned}
\widehat{\log({\rm BF}_{0,1})}=&\sum_{i=1}^{n}\log\left(\frac{\frac{1}{p}\sum_{j=1}^{p}f_{0}(y_{i}\vert \theta_{0}^{j})}{\frac{1}{p}\sum_{j=1}^{p}f_{1}(y_{i}\vert \theta_{1}^{j})}\right) +  n\log\left(\frac{\frac{1}{p}\sum_{j=1}^{p}g_{0}(\theta_{0}^{j})}{\frac{1}{p}\sum_{j=1}^{p}g_{1}(\theta_{1}^{j})}\right)\\
&+\frac{1}{2}\log(2\pi)\left(\vert \tilde{\theta}_{0}\vert- \vert \tilde{\theta}_{1}\vert\right)
     +
     \frac{1}{2}\log\left( \frac{\det(P_{1})}{\det(P_{0})}\right).
\end{aligned}
\end{align}
To compute $\widehat{\log({\rm BF}_{0,1})}$ in practice, we can use the computed log pointwise predictive density, and thus

\begin{align}\label{estlpd}
 \begin{aligned}
\widehat{\log({\rm BF}_{0,1})}=&\widehat{{\rm lpd}}_{0}- \widehat{{\rm lpd}}_{1}+  n\log\left(\frac{\frac{1}{p}\sum_{j=1}^{p}g_{0}(\theta_{0}^{j})}{\frac{1}{p}\sum_{j=1}^{p}g_{1}(\theta_{1}^{j})}\right)\\
&+\frac{1}{2}\log(2\pi)\left(\vert \tilde{\theta}_{0}\vert- \vert \tilde{\theta}_{1}\vert\right)
     +
     \frac{1}{2}\log\left( \frac{\det(P_{1})}{\det(P_{0})}\right).
\end{aligned}  
\end{align}
In some specific cases (see Section \ref{sec:examples}), we may suppose in equation  \eqref{logbif}, that the size of the free parameters, $\theta_{0}$, and $\theta_{1}$, is the same and that the first term on the right-hand side of equation \eqref{logbif} is negligible. Hence, one obtains the following decisional approach:

 \begin{align*}
     \log(\B_{0,1})=\frac{1}{2}\log\left( \frac{\det(P_{1})}{\det(P_{0})}\right).
 \end{align*}
Therefore, although our estimator might initially appear to be an artificial, added term, it actually recovers discrepancies between complex datasets with the same number of parameters. This can occur, for example, when the models differ in only one component, as illustrated by the cases studied in equation (\ref{formulanew}) below. An important fact of $\nbic$ regards its consistency. This is the content of the following result.

\begin{theorem}\label{thm:consistency}
Let us consider the indicator criterion $\nbic$ defined as
\begin{align*}
\nbic=-2 \log L(\hat{\theta}) + p \log n + \log|J|
\end{align*}
where $\log L(\hat{\theta})$ is the maximized log-likelihood, $p$ is the number of parameters in the model, $n$ is the sample size (number of observations), and $J$ is the observed Fisher Information matrix (negative Hessian), evaluated at the maximum likelihood estimator $\hat{\theta}$.

Then it is consistent. More precisely, given a true data-generating model $M_0$ (with $p_0$ parameters) and an overfitted competing model $M_1$ (with $p_1$ parameters, $p_1 > p_0$ such that $M_0 \subset M_1$), the probability of selecting the true model approaches one as the sample size increases:
\begin{align*}
\lim_{n \rightarrow \infty} \mathbb{P}(\Delta\nbic > 0) = 1,
\end{align*}
where $\Delta\nbic=\nbic(M_1) - \nbic(M_0)$.
\end{theorem}

\section{Examples}\label{sec:examples}
In this section, we analyze several models previously studied in \cite[Section 4]{perez2017scaled}, building upon their observations concerning the robust properties of the so-called Scaled Beta-2 distributions, which constitute a vast family of distributions used to tackle modeling scales in both hierarchical and non-hierarchical settings. Since we are interested in analyzing the formula \eqref{logbif}, let us consider the model discussed in \cite{perez2017scaled,gelman2006prior}:

\begin{align}\label{problem1}
X_{il} = \mu_i + \epsilon_{il} \quad \text{for } i=1, \dots, p \text{ groups and } l=1, \dots, r\, \text{observations per group.}    
\end{align}
This is a kind of Normal Hierarchical Model, specified as follows: The data data $x_{il}$ $l$-the observation in the $i$-th group, where the errors $\epsilon_{il} \sim N(0, \sigma^2)$ are i.i.d. real-valued Gaussian variables with mean null and variance $\sigma^2$. Notice that each group has the same mean value, and it is deterministic variable, but one may decide to choose even $\mu_{i}$ as a normal random variable. The parameter vector is $\theta = (\mu_1, \dots, \mu_p, \sigma^2)$, with dimension $p+1$. The log-likelihood function is:

\begin{align*}
l(\theta) = -\frac{pr}{2}\log(2\pi) - \frac{pr}{2}\log(\sigma^2) - \frac{1}{2\sigma^2}\sum_{i=1}^{p}\sum_{l=1}^{r}(X_{il} - \mu_i)^2.
\end{align*}
The maximum likelihood estimates (MLEs) are:
\begin{itemize}
    \item $\hat{\mu}_i = \bar{X}_{i.} = \frac{1}{r}\sum_{l=1}^{r}X_{il}$
    \item $\hat{\sigma}^2 = \frac{1}{pr}\sum_{i=1}^{p}\sum_{l=1}^{r}(X_{il} - \hat{\mu}_i)^2$
\end{itemize}
We calculate the Hessian matrix $\mathbf{H}(\theta)$, for general parameters before evaluating at the MLE. Hence, one gets

\begin{align}\nonumber
&\frac{\partial l}{\partial \mu_j} = \frac{1}{\sigma^2}\sum_{l=1}^{r}(X_{jl} - \mu_j),\\\nonumber
&\frac{\partial l}{\partial(\sigma^2)} = -\frac{pr}{2\sigma^2} + \frac{1}{2(\sigma^2)^2}\sum_{i=1}^{p}\sum_{l=1}^{r}(X_{il} - \mu_i)^2,\\\nonumber
&\frac{\partial^2 l}{\partial \mu_j^2} = -\frac{r}{\sigma^2}, \quad \frac{\partial^2 l}{\partial \mu_j \partial \mu_k} = 0 \quad (\text{for } j \neq k),\\\nonumber
&\frac{\partial^2 l}{\partial(\sigma^2)^2} = \frac{pr}{2(\sigma^2)^2} - \frac{1}{(\sigma^2)^3}\sum_{i=1}^{p}\sum_{l=1}^{r}(X_{il} - \mu_i)^2,\\\label{formu5}
&\frac{\partial^2 l}{\partial \mu_j \partial(\sigma^2)} = -\frac{1}{(\sigma^2)^2}\sum_{l=1}^{r}(X_{jl} - \mu_j).
\end{align}
We notice that \eqref{formu5} is not null for a general parameter $\mu_j$, and tha the Hessian matrix for a general $\theta$ is not block diagonal. However, when we evaluate it at the MLE, $\hat{\theta} = (\hat{\mu}_1, \dots, \hat{\mu}_p, \hat{\sigma}^2)$, the off-diagonal terms vanish. Specifically, the cross-derivative terms become:
\[
\left. \frac{\partial^2 l}{\partial \mu_j \partial(\sigma^2)} \right|_{\theta=\hat{\theta}} = -\frac{1}{(\hat{\sigma}^2)^2}\sum_{l=1}^{r}(X_{jl} - \hat{\mu}_j) = -\frac{1}{(\hat{\sigma}^2)^2} \times 0 = 0
\]
And the bottom-right corner term simplifies to:
\[
\left. \frac{\partial^2 l}{\partial(\sigma^2)^2} \right|_{\theta=\hat{\theta}} = \frac{pr}{2(\hat{\sigma}^2)^2} - \frac{pr\hat{\sigma}^2}{(\hat{\sigma}^2)^3} = -\frac{pr}{2(\hat{\sigma}^2)^2}
\]
Therefore, the Hessian evaluated at the MLE is block diagonal:
\[
\mathbf{H}(\hat{\theta}) = 
\begin{pmatrix}
-\frac{r}{\hat{\sigma}^2}\mathbf{I}_{p \times p} & \mathbf{0} \\
\mathbf{0}^T & -\frac{pr}{2(\hat{\sigma}^2)^2}
\end{pmatrix}
\]
By definition, the Observed Fisher Information matrix is $\mathbf{I}(\hat{\theta}) = -\mathbf{H}(\hat{\theta})$:
\begin{align*}
\mathbf{I}(\hat{\theta}) = 
\begin{pmatrix}
\frac{r}{\hat{\sigma}^2}\mathbf{I}_{p \times p} & \mathbf{0} \\
\mathbf{0}^T & \frac{pr}{2(\hat{\sigma}^2)^2}.
\end{pmatrix}
\end{align*}
The determinant is the product of the determinants of the diagonal blocks:
\begin{align*}
\det(\mathbf{I}(\hat{\theta})) = \left(\frac{r}{\hat{\sigma}^2}\right)^p \times \frac{pr}{2(\hat{\sigma}^2)^2}  
\end{align*}
and the log-determinant is:
\begin{align*}
\log\det(\mathbf{I}(\hat{\theta})) &= p\log\left(\frac{r}{\hat{\sigma}^2}\right) + \log\left(\frac{pr}{2(\hat{\sigma}^2)^2}\right) \\
&= p\log(r) - p\log(\hat{\sigma}^2) + \log(p) + \log(r) - \log(2) - 2\log(\hat{\sigma}^2) \\
&= (p+1)\log(r) + \log(p) - \log(2) - (p+2)\log(\hat{\sigma}^2).
\end{align*}
\subsection*{Comparison between BIC and the modified criterion}\label{comparison}
The key distinction between the standard BIC and the proposed criterion lies in the structure of their penalty terms.  
For this example \eqref{problem1}, the total number of parameters is $p_{\text{total}} = p + 1$.  
Assuming $n = r$ as the effective sample size (each $\mu_i$ is estimated from $r$ observations), the classical BIC penalty is
\begin{align}\label{penaltyBIC}
\text{Penalty}_{\text{BIC}} = p_{\text{total}}\log n = (p + 1)\log(r).
\end{align}

This formulation ignores the internal structure of the model and the amount of information available for each parameter.  
By contrast, incorporating the log-determinant of the Fisher information matrix yields
\begin{align}\label{penaltyBH}
\text{Penalty}_{\text{HES}} \coloneqq (p + 1)\log(r) + \log(p) - \log(2) - (p + 2)\log(\hat{\sigma}^2),
\end{align}
which retains the BIC term as its first component but adds three meaningful adjustments:
\begin{enumerate}
    \item \textbf{$+\log(p)$:} introduces a group-dependent penalty, capturing the hierarchical dimension of the model;
    \item \textbf{$-\log(2)$:} a constant scaling adjustment;
    \item \textbf{$-(p + 2)\log(\hat{\sigma}^2)$:} a data-dependent term reflecting the precision of model fit.
\end{enumerate}

When $\hat{\sigma}^2$ is small (indicating a tight fit), $\log(\hat{\sigma}^2)$ becomes large and negative, so the penalty term $-(p + 2)\log(\hat{\sigma}^2)$ grows positive and substantial.  
This adaptively penalizes models that may overfit by explaining random noise rather than true signal. For instance, consider two models $H_{0}$, and $H_{1}$ as in \eqref{problem1} with $r_{0}$, and $r_{1}$ observations per group, respectively. Assume both models have exactly $p+1$ parameters, and differ only in their variance terms $\sigma_{0}^{2}$, and $\sigma_{1}^{2}$, respectively. For simplicity, suppose the first term on the right-hand side of \eqref{logbif} is negligible. Hence,

\begin{align}\label{formulanew}
\log(\B_{0,1})=(p+1)\log(r_{1}/r_{0})+ (p+2)\log(\hat{\sigma}_{0}^{2}/\hat{\sigma}_{1}^{2}),  
\end{align}
where $\hat{\sigma}_{0}^{2},\hat{\sigma}_{1}^{2}$ are the empirical estimators of $\sigma_{0}^{2}$, and $\sigma_{1}^{2}$, respectively. Suppose that $0<\hat{\sigma}_{0}^{2}<\hat{\sigma}_{1}^{2}$, and that $r_{1}<r_{0}$. Then model $H_{0}$ contains more noise variables and, under the proposed BIC$\_$HES criterion, is effectively penalized. In contrast, the standard BIC does not guarantee effective discrimination between models that include redundant noise variables.

\section{Simulations}\label{simulations}
In this section, we present the simulation experiments conducted to assess the performance of the modified Bayesian Information Criterion (BIC$\_$HES), which incorporates the logarithm of the determinant of the Hessian matrix of the parameters. The main objective is to analyze how this modification affects model selection under different scenarios, in comparison with traditional criteria such as the standard BIC, Akaike Information Criterion (AIC), corrected AIC (AICc), the Corrected Akaike Information Factor (CAIF), and the Information Complexity Criterion (ICOMP) \cite{akaike1973,HURVICHTSAI1989,bozdogan1987model,bozdogan1988icompcriterion}. \\
In what follows, within the figures, we use the term metric to refer to the specific type of Bayesian factor employed for data analysis. The acronym PVCSVM stands for the percentage of times a given metric selects the correct model. This abbreviation derives from the Spanish phrase “porcentaje de veces que un criterio selecciona el verdadero modelo.” Additionally, we define BIC\_HES\_SP as the modified version of BIC\_HES obtained by subtracting the term $p\log(n)$, where $n$ denotes the sample size and, $p$ the number of parameters in the model. This adjustment is motivated by the observation that BIC\_HES, may, in certain cases, include the penalization $p\log(n)$ twice.

\subsection{Simulation Study for the Linear Model}\label{sec:linearcase}
This section describes the simulation design for the linear model. All simulations were carried out in R \cite{rsoftware}, version 4.3.2, on a Linux workstation equipped with 1 TB of RAM and 32 Intel Xeon 2.7 GHz cores. The simulation approach follows \cite{neath1997regression}, where a true data-generating model is defined, and additional noise variables are introduced to refit the model and evaluate the performance of various selection criteria. A total of 200 replications were conducted, generating 1,096,800 fitted models. The performance of each criterion was assessed using the empirical frequency with which it identified the true data-generating model.

The general model is defined as:

\begin{equation}
\label{EQlineal}
Y_i = \beta_0 + \sum_{k=2}^{10} \beta_k X_{ik} + \sum_{l=1}^{10} \beta_l X_{il} + \varepsilon_i, \quad i = 1, \ldots, n,\; k = 2,4,\ldots,10,\; l = 1,\ldots,10.
\end{equation}

where the noise variables correspond to indices $l = 1, \ldots, 10$.

The simulation procedure was implemented as follows:

\begin{enumerate}[label=\alph*)]
\item Generate $k$ independent regressor variables, where $k = 2, 4, 6, 8, 10$. Each variable follows a normal distribution with mean 10 and standard deviation 1.
\item Generate $k$ regression parameters, each drawn from a normal distribution with mean 10 and standard deviation 2.
\item Generate normally distributed errors with mean 0 and standard deviation 1.
\item Compute the response variables using Equation (\ref{EQlineal}) for sample sizes $n = 10, 15, 20, 50, 100, 500$, and varying numbers of regressors $X_k$, excluding the noise variables.
\item Generate $l = 1, \ldots, 10$ additional noise regressors from a normal distribution with mean 3 and standard deviation 3.
\item Append the noise variables to the models from step d), ensuring that their number does not exceed the number of original regressors.
\item Compute the information criteria (BIC, AIC, AICc, BIC\_HES, BIC\_HES\_SP, and CAIF).
\end{enumerate}

\begin{figure}[htbp]
	\centering
	\includegraphics[width=0.70\textwidth]{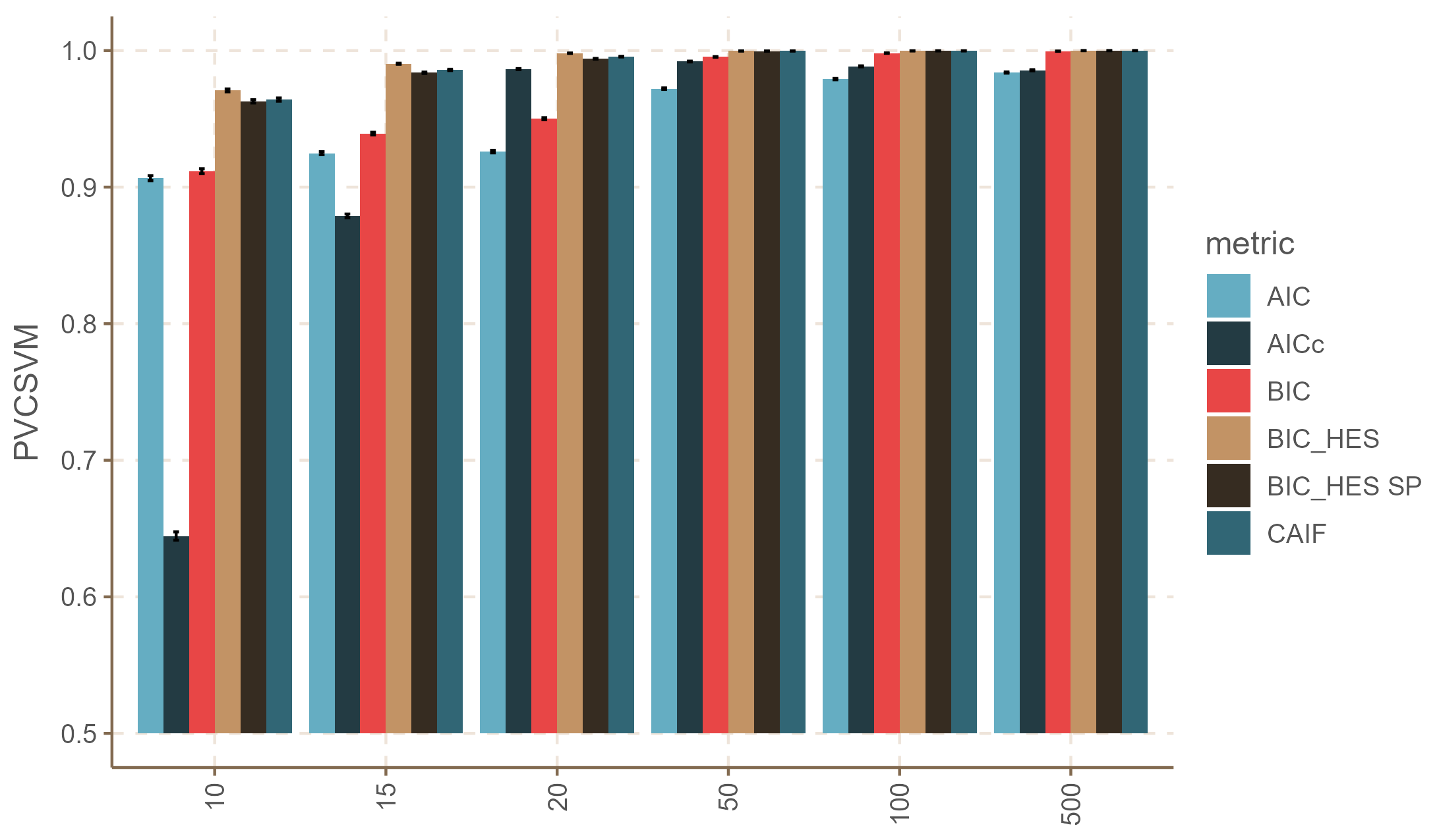}
	\caption{Information criteria vs sample size}
	\label{fig:criterioslinealnobs}
\end{figure}

In Figure \ref{fig:criterioslinealnobs}, the BIC\_HES and BIC\_HES\_SP criteria exhibit superior performance in identifying the true model, particularly for small sample sizes. This indicates that incorporating the Hessian matrix introduces a more appropriate penalization term in such conditions. Despite their popularity, AIC and AICc perform less effectively. As the sample size increases, all criteria tend to converge, suggesting that the choice of information criterion becomes less critical when sufficient data are available. Under small-sample conditions, BIC\_HES and BIC\_HES\_SP appear to provide the most reliable results.

Figure \ref{fig:BICvsBICDIO} presents a direct comparison between BIC and BIC\_HES across different sample sizes. The horizontal axis shows the BIC values, while the vertical axis corresponds to the BIC\_HES values. The strong linear association between these two criteria suggests that the inclusion of the logarithmic determinant of the Hessian matrix contributes meaningful information to the selection process.
\begin{figure}[htbp]
	\centering
		\includegraphics[width=0.60\textwidth]{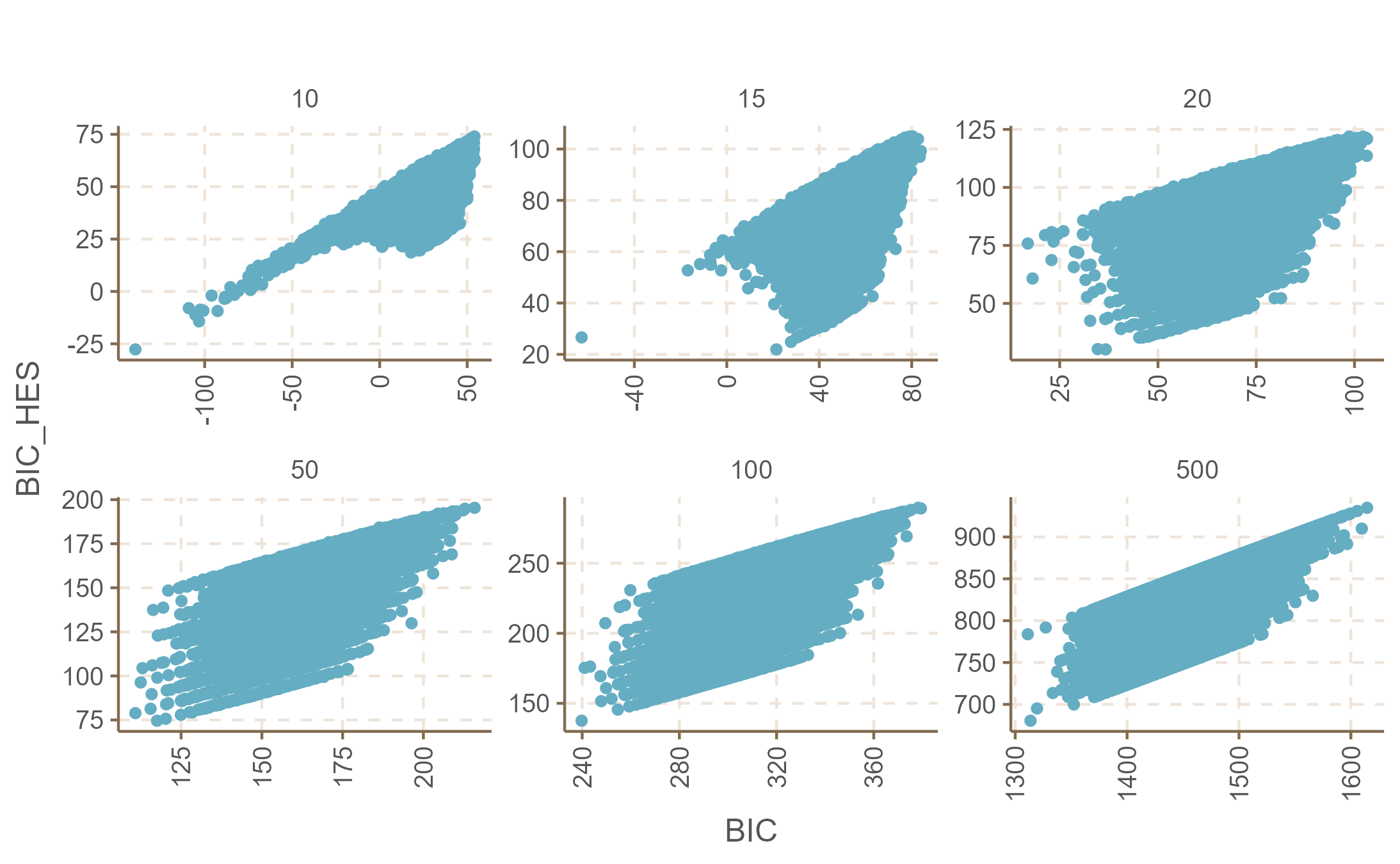}
\caption{Comparison between BIC\_HES y BIC.}	
 \label{fig:BICvsBICDIO}
	\end{figure}

A similar comparison is depicted in Figure \ref{fig:AICvsBICDIO}, replacing BIC with AIC. The observed relationship between BIC\_HES and AIC reinforces that the additional term in BIC\_HES assists in distinguishing between competing models.
\begin{figure}[htbp]
	\centering
		\includegraphics[width=0.60\textwidth]{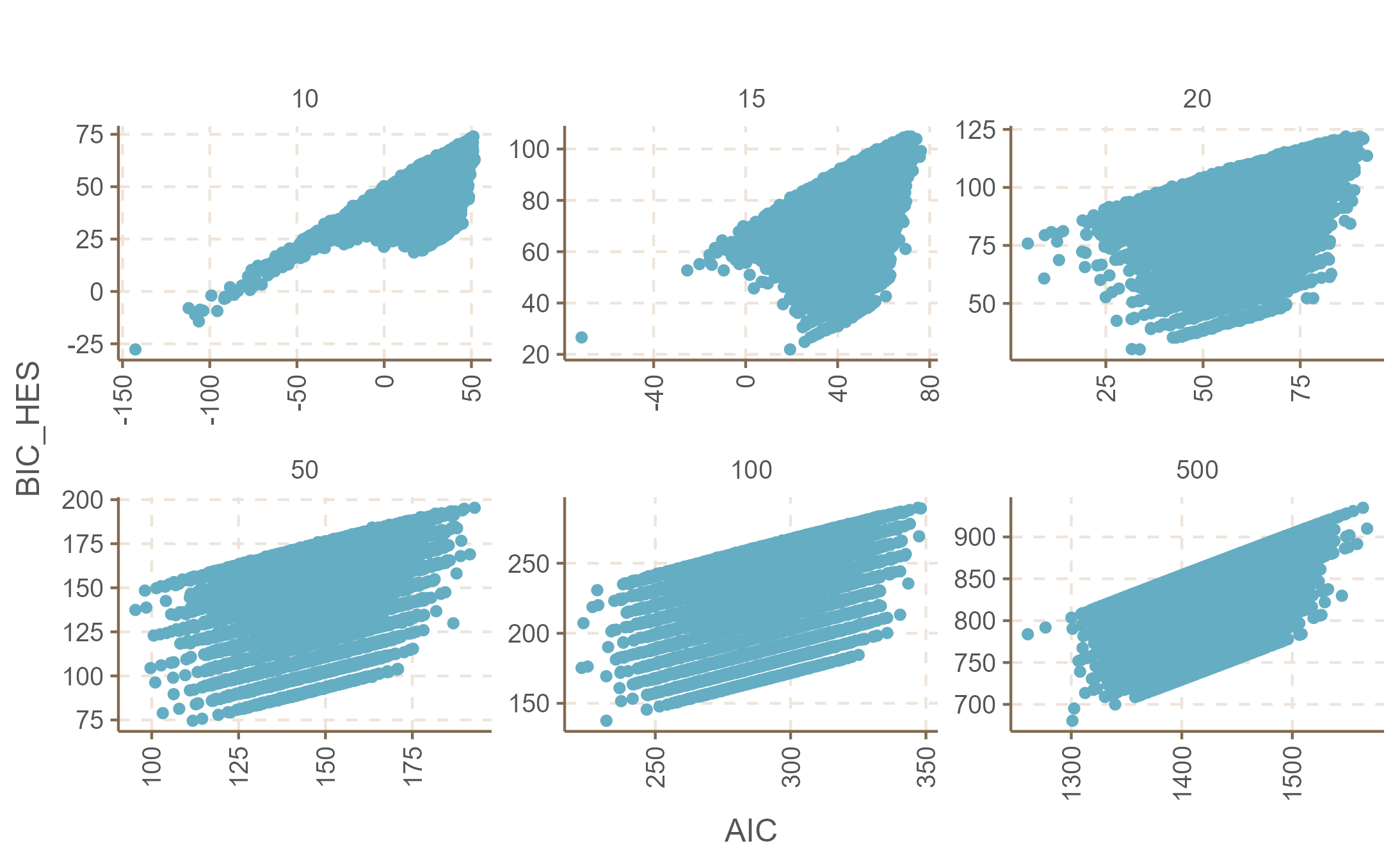}
	\caption{Comparison between BIC\_HES y AIC.}
 \label{fig:AICvsBICDIO}
\end{figure}

Figure \ref{fig:CRITERIOSVARCONT} compares the selection frequencies of the different criteria as a function of the number of added noise variables. BIC\_HES selects the true model more frequently than BIC when noise variables are introduced. However, when ten noise variables are included, both criteria yield nearly identical performance. Figure \ref{fig:CRITERIOSVARCONTNOBS} summarizes these results jointly by number of added noise variables and sample size. Once again, BIC\_HES demonstrates a higher probability of selecting the true data-generating model.

\begin{figure}[htbp]
	\centering
		\includegraphics[width=0.70\textwidth]{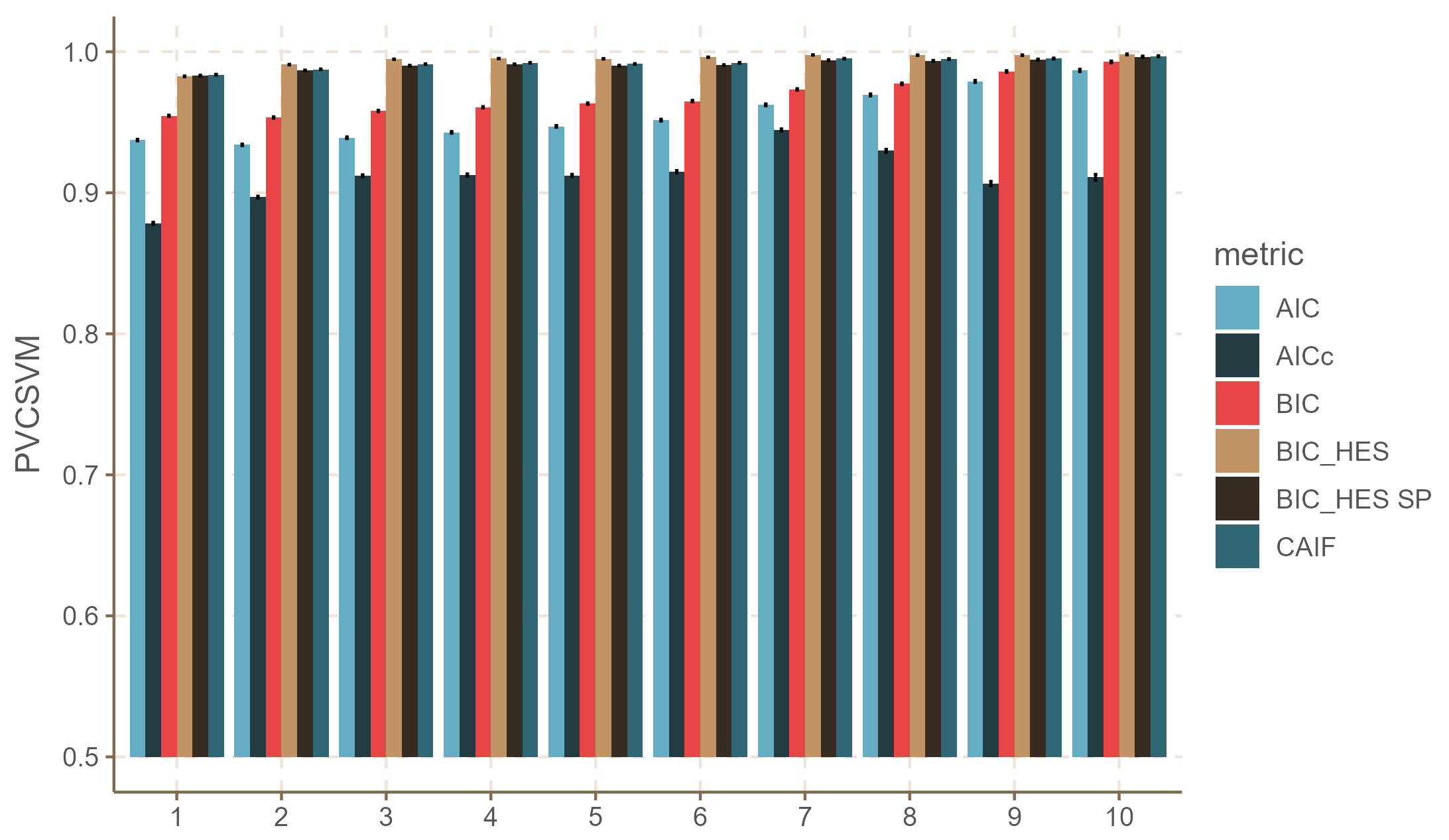}
		\caption{Information criteria vs number of noise variables}
		\label{fig:CRITERIOSVARCONT}
	\end{figure}
	
	\begin{figure}[h]
	\centering
		\includegraphics[width=1.0\textwidth]{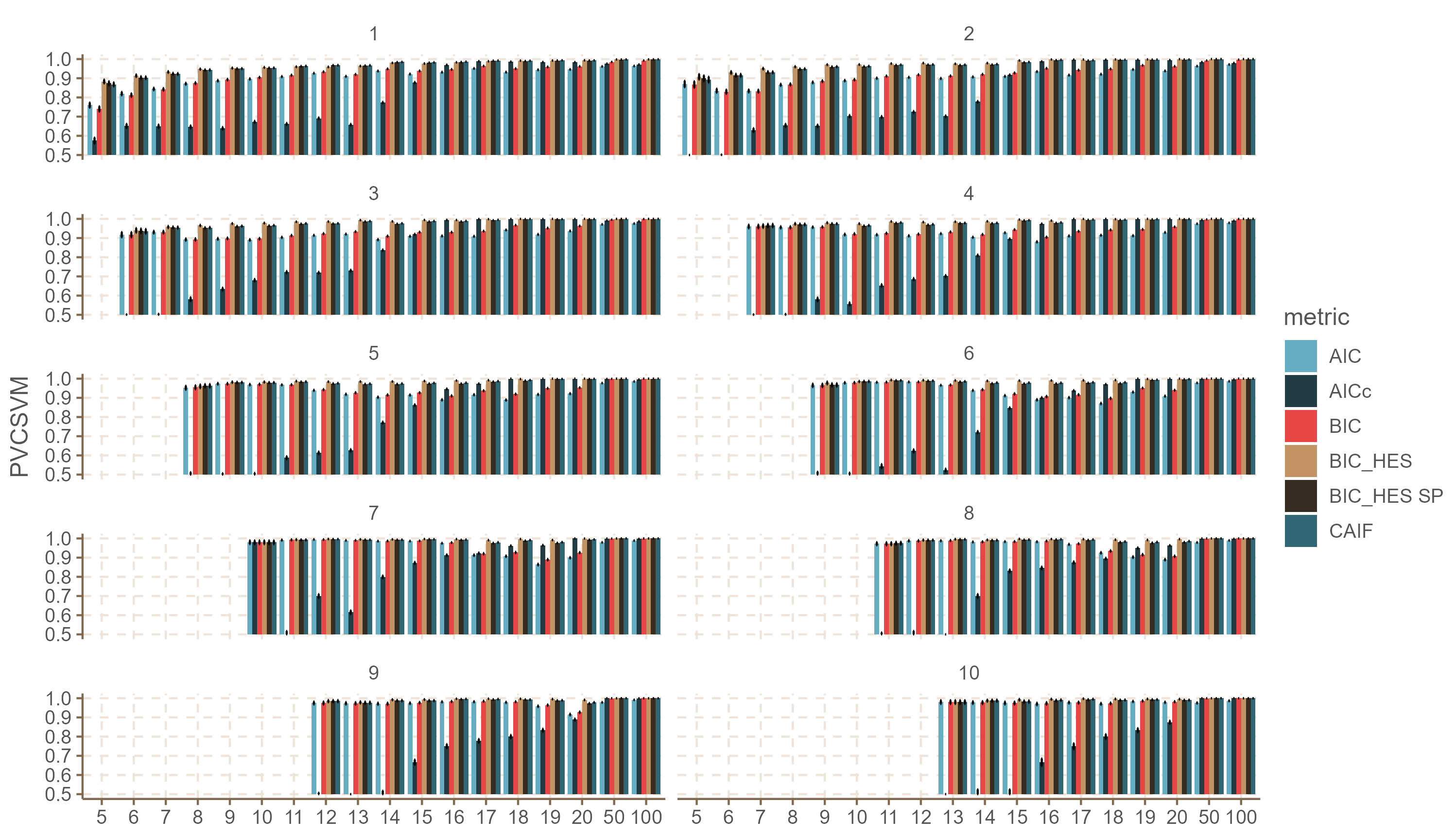}
		\caption{Information criteria by number of noise variables and sample size}
	\label{fig:CRITERIOSVARCONTNOBS}
\end{figure}
Overall, the results show a clear association between BIC and the logarithmic determinant of the Hessian matrix of the parameters. This component can thus be incorporated into the BIC expression to enhance model discrimination in multiple linear regression frameworks.

\subsection{Simulations for the Linear Mixed Model}\label{sec:mixedcase}

This section reports simulations for the linear mixed model, defined as:

\begin{equation}
\label{EQMIXTO}
y_{ij} = \beta_0 + \beta_1 X_{1ij} + \beta_2 X_{2ij} + \beta_3 X_{3ij} + u_{0i} + u_{1i} Z_{1ij} + u_{2i} Z_{2ij} + u_{3i} Z_{3ij} + \varepsilon_{ij}.
\end{equation}

where:

\begin{itemize}
\item $y_{ij}$ denotes the response for observation $j$ within group $i$;
\item $\beta_0$ is the fixed intercept;
\item $\beta_1, \beta_2, \beta_3$ are fixed-effect coefficients;
\item $u_{0i} \sim N(0, \sigma^2_{u0})$ represents the random intercept for group $i$;
\item $u_{1i}, u_{2i}, u_{3i} \sim N(0, \sigma^2_{bk})$ denote random-effect coefficients;
\item $\varepsilon_{ij} \sim N(0, \sigma^2_\varepsilon)$ represents the residual error term.
\end{itemize}

The simulation design was as follows:

\begin{enumerate}[label=\alph*)]
\item Generate regressors for fixed effects from a normal distribution with mean 2 and standard deviation 3.
\item Generate parameters for fixed effects from a normal distribution with mean 2 and standard deviation 3.
\item Generate global errors from a normal distribution with mean 0 and standard deviation 1.
\item Generate regressors for random effects from a normal distribution with mean 3 and standard deviation 3.
\item Generate parameters for random effects from a normal distribution with mean 3 and standard deviation 1.
\item Generate group-level random errors from a normal distribution with mean 0 and standard deviation 2.
\item Compute the response variables using Equation (\ref{EQMIXTO}) for sample sizes $n = 20, 50, 100, 500$ and group sizes $N = 5, 10, 20$.
\item Add noise regressors for fixed effects (mean 3, SD 3).
\item Add noise regressors for random effects (mean 3, SD 3).
\item Introduce the noise variables into the models, ensuring they do not exceed the number of original regressors.
\item Compute the information criteria (BIC, AIC, AICc, BIC\_HES, BIC\_HES\_SP, and CAIF).
\end{enumerate}

The simulations were executed in R \cite{rsoftware} using version 4.3.2 on a Linux system with 1 TB of RAM and 32 Intel Xeon 2.7 GHz cores. The \texttt{GLMMadaptive} package \cite{rizopoulos2023} was employed for model estimation. This package fits mixed-effects models for clustered data where integration over random effects cannot be solved analytically, using adaptive Gauss-Hermite quadrature. It also provides the Hessian matrix of the estimated parameters, which was essential for this analysis. A total of five replications were performed, generating 8,825 models.

Figure \ref{fig:criteriosVSN_obs} illustrates that BIC\_HES detects the true model more frequently than other criteria, reaching comparable performance at larger sample sizes. For small samples, BIC\_HES continues to outperform competing criteria. Figures \ref{fig:bicngrupos} to \ref{fig:noise_adicionadasNobs} present results by group size and by the number of added noise variables for both fixed and random effects. Across all scenarios, BIC\_HES consistently identifies the true data-generating model more often than the conventional BIC, particularly for small sample sizes or when noise variables are present. These findings underscore the robustness of the modified criterion in mixed-effects model frameworks.
\begin{figure}[htbp]
		\centering
		\includegraphics[width=0.90\textwidth]{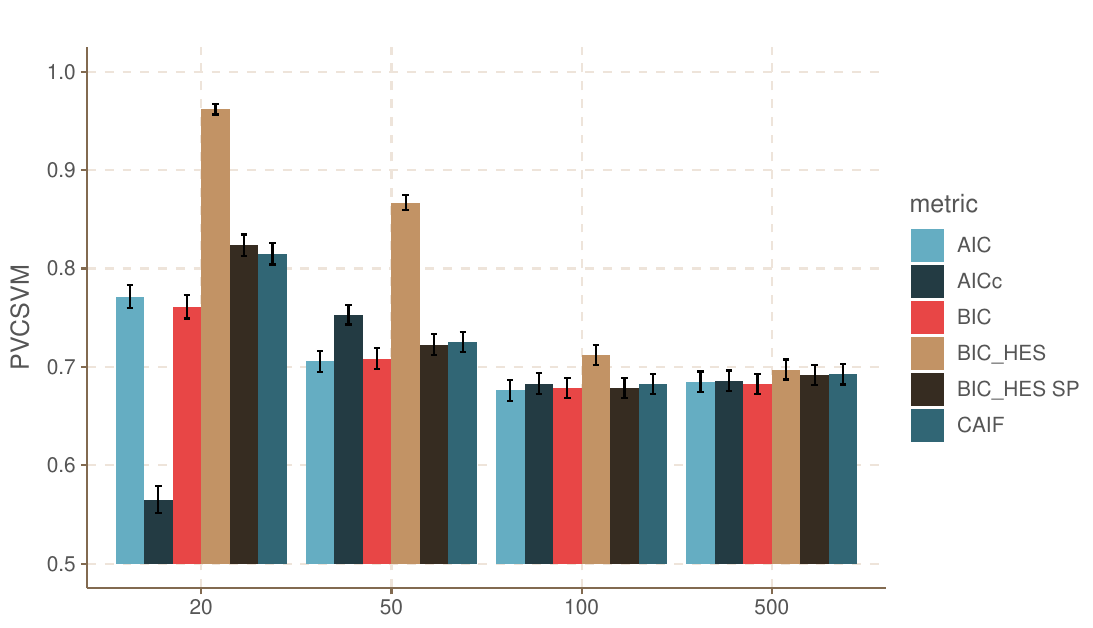}
	\caption{Information criteria vs sample size}
	\label{fig:criteriosVSN_obs}
\end{figure}
\begin{figure}[htbp]
    \centering
    \includegraphics[width=0.90\textwidth]{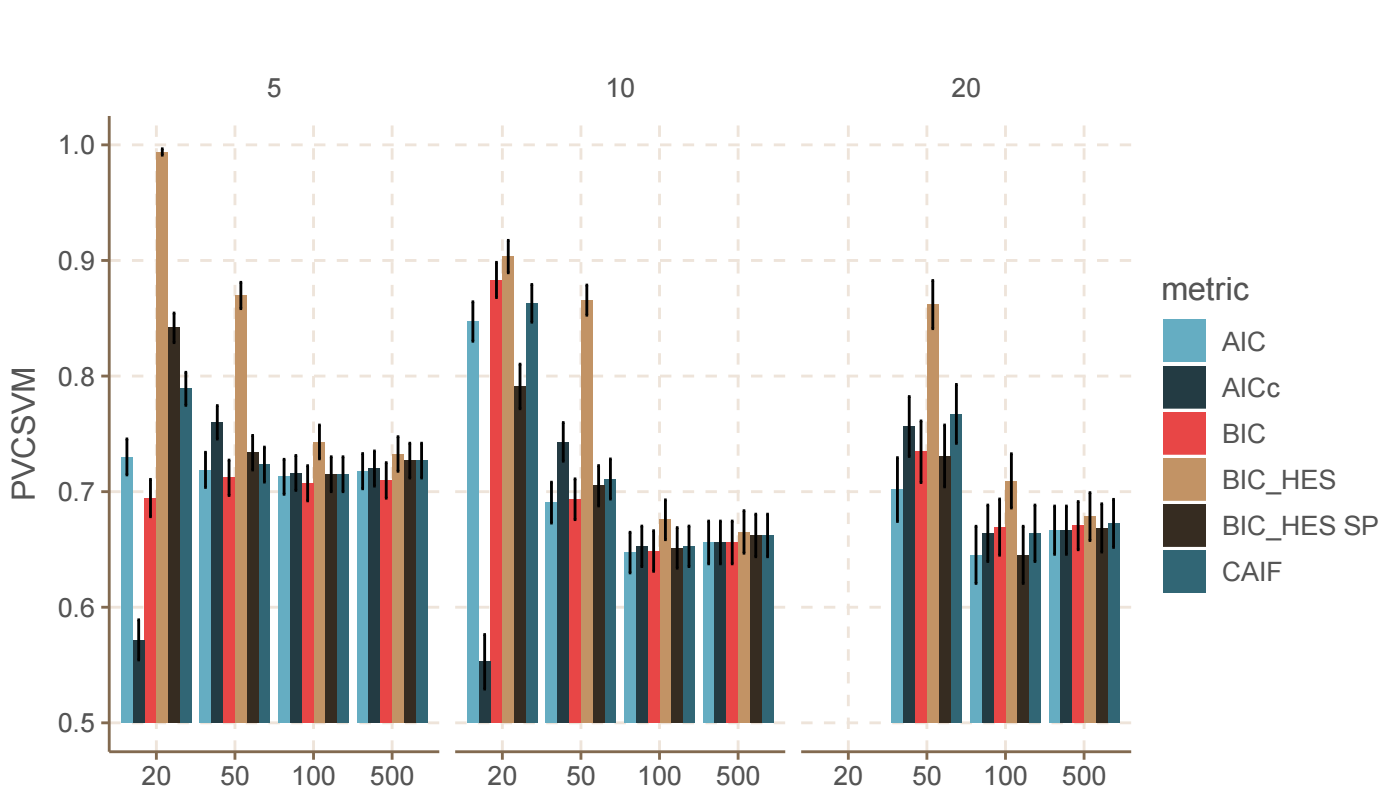}
    \caption{Information criteria by group size and sample size}
    \label{fig:bicngrupos}
\end{figure}
\begin{figure}[htbp]
    \centering
    \includegraphics[width=0.90\textwidth]{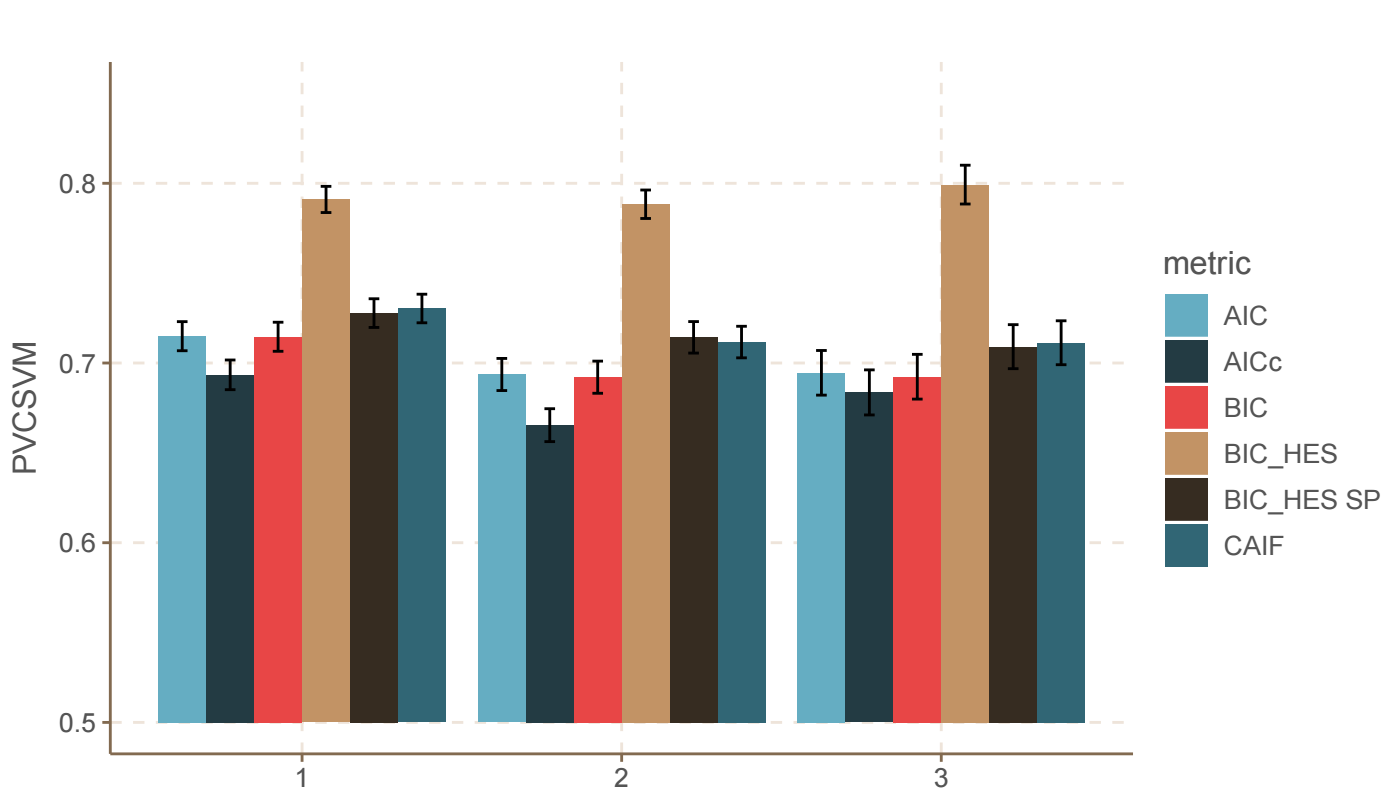}
    \caption{Information criteria by number of added fixed-effect variables}
    \label{fig:fixed_adicionadas}
\end{figure}
\begin{figure}[htbp]
    \centering
    \includegraphics[width=0.90\textwidth]{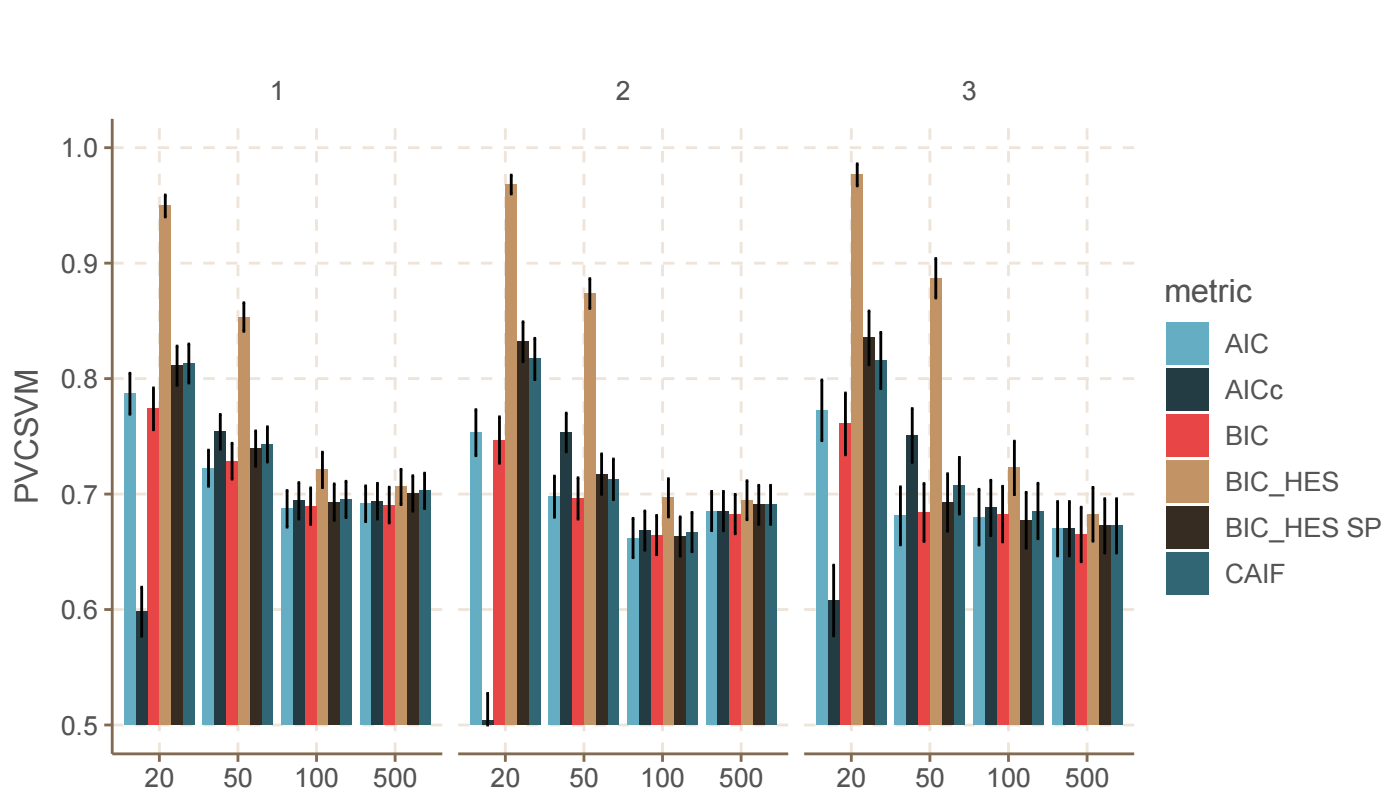}
    \caption{Information criteria by sample size and number of added fixed-effect variables.}
    \label{fig:fixed_adicionadas_Nobs}
\end{figure}
\begin{figure}[htbp]
    \centering
    \includegraphics[width=0.90\textwidth]{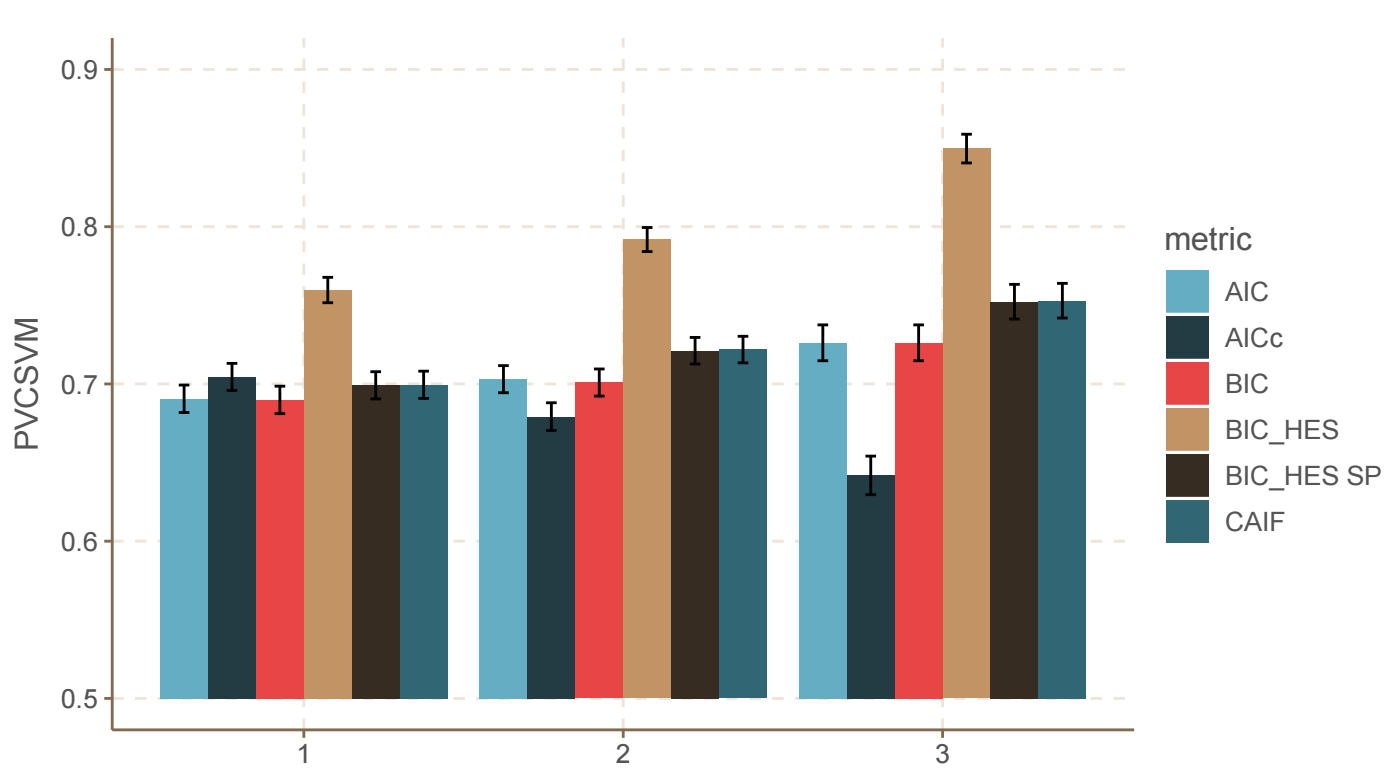}
    \caption{Information criteria by number of added random-effect noise variables.}
    \label{fig:noise_adicionadas}
\end{figure}

\begin{figure}[htbp]
    \centering
    \includegraphics[width=0.90\textwidth]{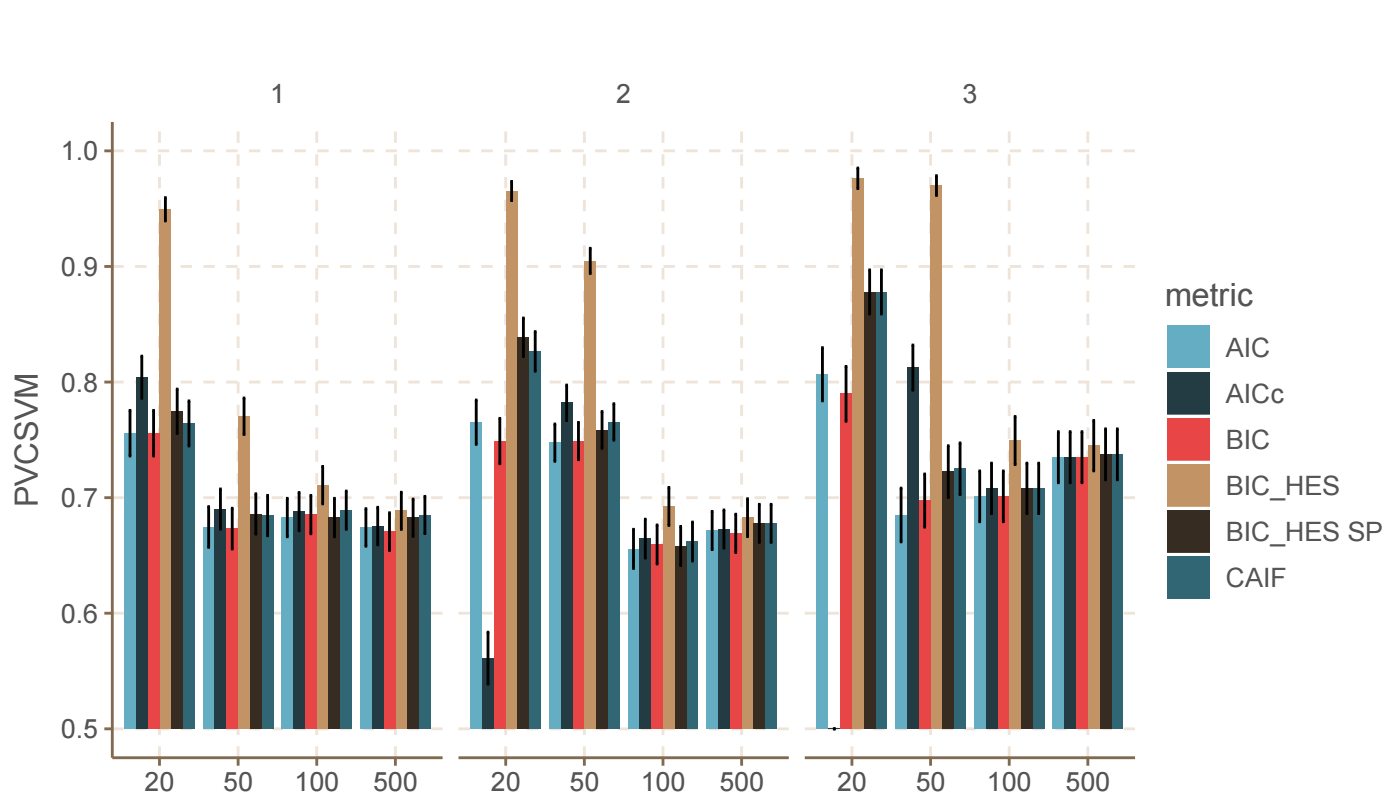}
    \caption{Information criteria by sample size and number of added random-effect noise variables}
    \label{fig:noise_adicionadasNobs}
\end{figure}

\subsection{Simulations for Binomial and Poisson Models}\label{sec:binomialpoisson}

This section presents simulations aimed at evaluating the effectiveness of several information criteria, including AIC, BIC, AICc, BIC\_HES, BIC\_HES\_SP, CAIF, and ICOMP. The procedure, based on \cite{debruine2021}, involved fitting three theoretical models: a generalized Gaussian mixed model, a logistic model, and a Poisson model. The parameters for each model were defined as follows:

\begin{itemize}
\item Number of subjects: 20, 50, 100;
\item Number of in-group stimuli: 2, 5, 10;
\item Number of out-group stimuli: 2, 5, 10;
\item Correlation between intercept and slope: $-0.7, -0.5, 0, 0.5, 0.7$.
\end{itemize}

Model specifications:
\begin{align*}
\text{Gaussian mixed model:} & \quad \beta_0 = 800,\; \beta_1 = 50, \\
\text{Logistic mixed model:} & \quad \beta_0 = 0.5,\; \beta_1 = 1, \\
\text{Poisson mixed model:} & \quad \beta_0 = 5,\; \beta_1 = 4.
\end{align*}

The simulation protocol was structured as follows:

\begin{enumerate}
\item \textbf{Data simulation:} Synthetic datasets were generated from predefined parameter combinations for each model type, ensuring a controlled environment for evaluating the selection criteria.
\item \textbf{Fitting the true model:} The theoretical models (Gaussian, logistic, Poisson) were fitted, and multiple criteria BIC, AIC, AICc, BIC\_HES, and BIC\_HES\_SP were applied.
\item \textbf{Fitting the misspecified model:} Alternative models were also fitted to examine whether the criteria could discriminate between true and incorrect specifications.
\item \textbf{Criterion comparison:} The values of BIC\_HES for the correct (M1) and incorrect (M2) models were compared to determine the criterion’s capacity to identify the correct model.
\item \textbf{Decision rule:}
\begin{itemize}
\item If ${\rm BIC}\_{\rm HES}_{M2} > {\rm BIC}\_{\rm HES}_{M1}$, the criterion fails to detect the true model.
\item If ${\rm BIC}\_{\rm HES}_{M2} \leq {\rm BIC}\_{\rm HES}_{M1}$, the criterion successfully identifies the true model.
\end{itemize}
\item \textbf{Repetition:} The process was repeated under multiple parameter combinations for all three model families to ensure robustness across different configurations.
\end{enumerate}

Simulations were performed using R version 4.3.2 \cite{rsoftware} on a Linux server with 512 GB of RAM and 16 Intel Xeon Gold 5320 processors at 2.2 GHz. The \texttt{GLMMadaptive} package \cite{rizopoulos2023} was employed, as it estimates mixed-effects models for grouped data and provides the Hessian matrix of the parameters used in this analysis. Twenty replications were conducted, generating and evaluating 124,273 models.

Figure \ref{fig:grafico_model_type} shows the proportion of times each criterion selected the true model. BIC, BIC\_HESS, and BIC\_HESS\_SP exhibit the highest effectiveness in identifying the true model in binomial and Poisson mixed models, owing to their appropriate complexity penalization and efficient handling of random effects. In contrast, AIC, AICc, CAIF, and ICOMP display more variable and frequently inferior performance, underscoring the need for criteria that properly balance goodness of fit and model complexity in these contexts.
\begin{figure}[htbp]
	\centering
		\includegraphics[width=0.70\textwidth]{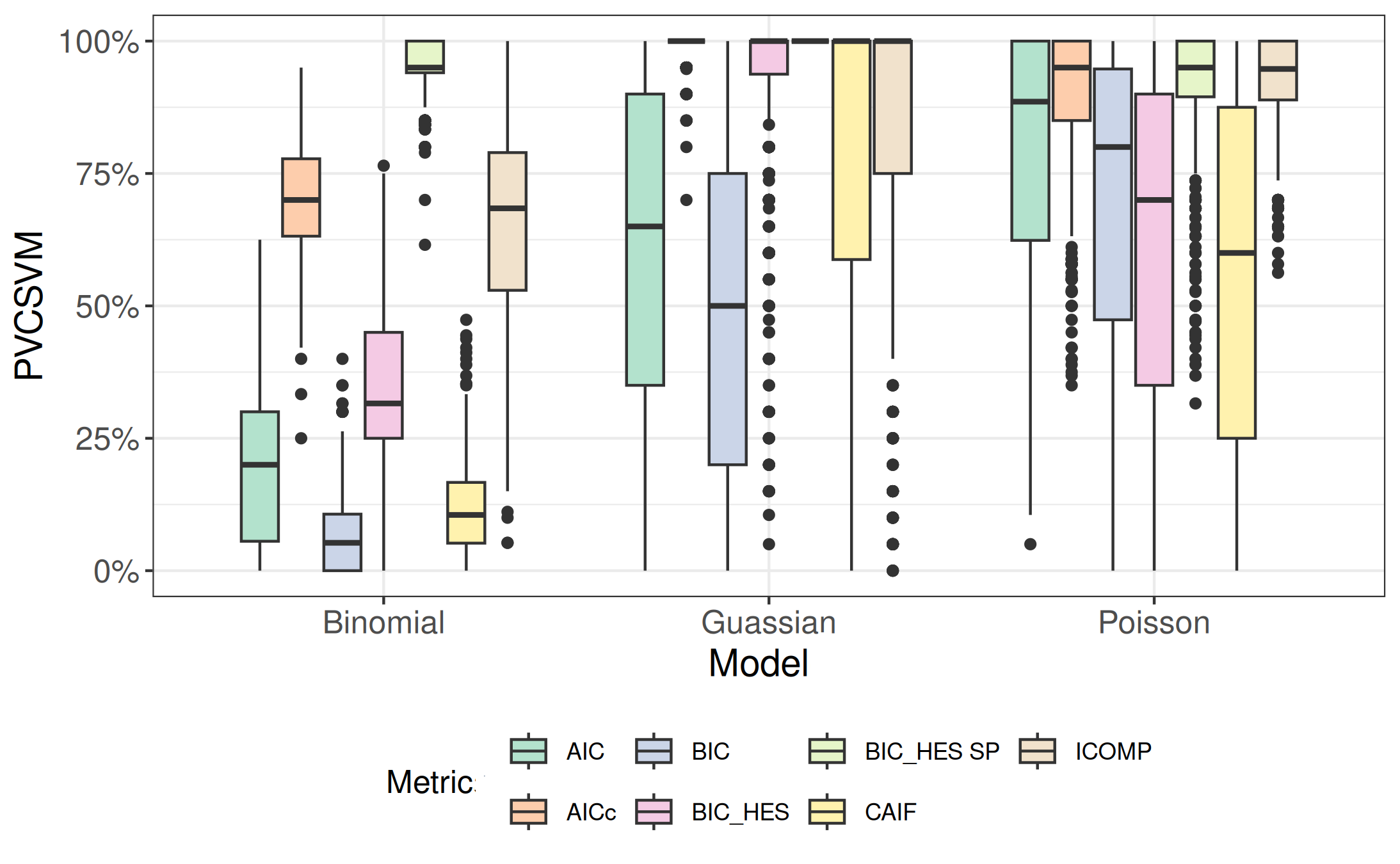}
				\caption{Information criteria for binomial and Poisson models}
    \label{fig:grafico_model_type}
\end{figure}

Figure \ref{fig:correlacion} indicates that BIC and BIC\_HESS are the most robust criteria across Gaussian, binomial, and Poisson mixed-model simulations, regardless of the correlation between random effects. Although AIC and AICc perform competitively under certain settings, they are more sensitive to extreme correlation values ($\rho$). CAIF and ICOMP yield less consistent results, suggesting limited suitability when correlations among random effects are pronounced. These findings emphasize the importance of adopting information criteria capable of effectively addressing the additional complexity introduced by correlated random effects in mixed-effects modeling.
\begin{figure}[htbp]
	\centering
		\includegraphics[width=0.70\textwidth]{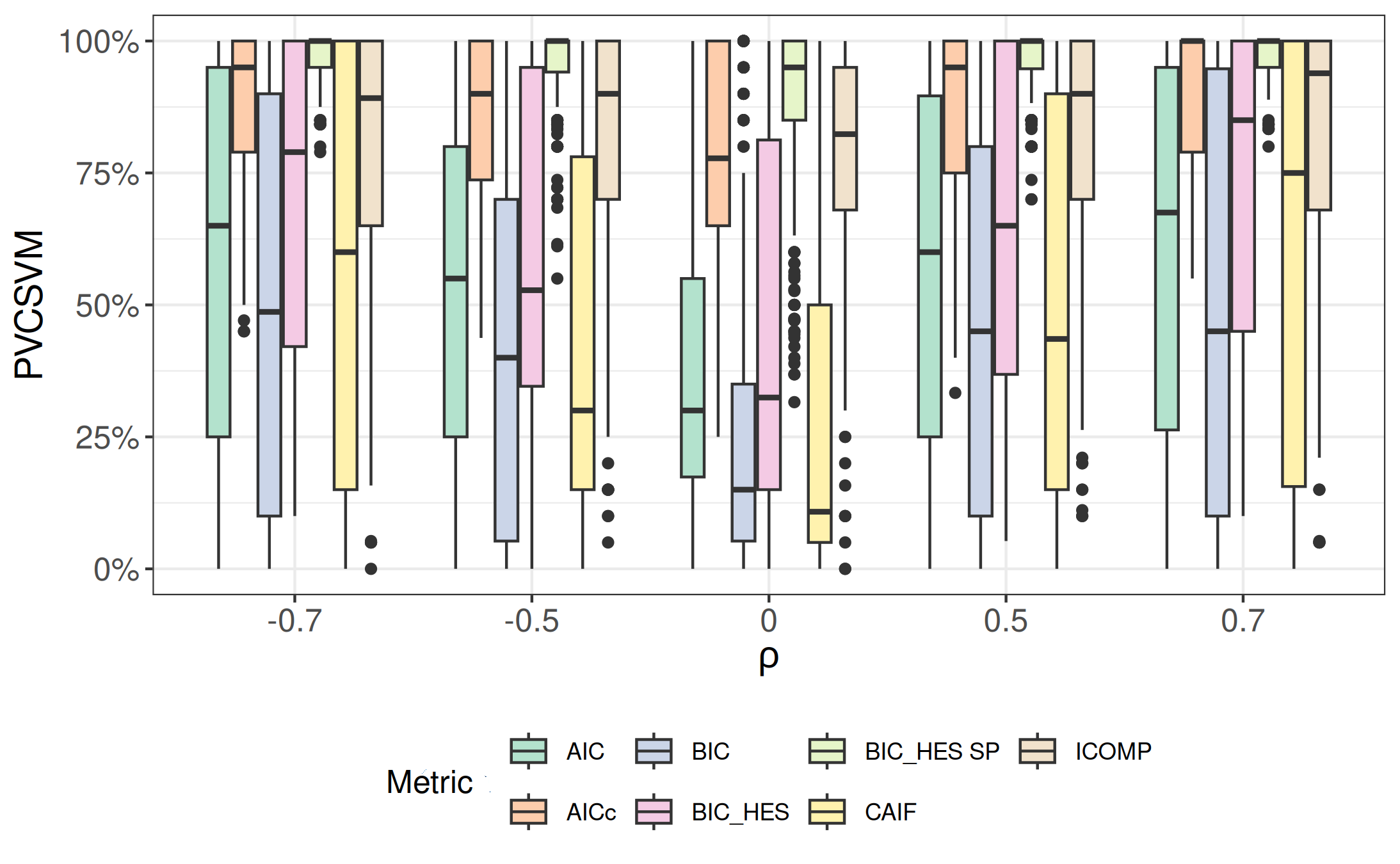}
				\caption{Selection percentage by correlations among random effects}	
    \label{fig:correlacion}
\end{figure}

\section{Proofs}\label{sec:proofs}
In this section, we provide the proofs of Corollary \ref{1cor}, and  Theorem \ref{betterthm}.

 \begin{proof}[Proof of Corollary \ref{1cor}]

 \begin{align*}
     \P(y\vert H_{i})&=\int_{\Theta_{i}}\P(y\vert H_{i},\theta)g_{i}(\theta)d\theta\\
     &=\int_{\Theta_{i}}f_{i}(y\vert \theta)g_{i}(\theta)d\theta\\
     &=\int_{\Theta_{i}}{\rm exp}\left( \log \left(f_{i}(y\vert \theta)g_{i}(\theta)\right)\right)d\theta
 \end{align*}
 We can now expand $\log \left(f_{i}(y\vert \theta)g_{i}(\theta)\right)$ about $\tilde{\theta}$.  Thus, we  can make the following second order expansion about  $\theta=\tilde{\theta}_{i}$ as

\begin{align*}
    \log \left(f_{i}(y\vert \theta)g_{i}(\theta)\right)= \log \left(f_{i}(y\vert \tilde{\theta}_{i})g_{i}(\tilde{\theta}_{i})\right) &+ (\theta-\theta_{i})\cdot\nabla Q(\tilde{\theta}_{i})+\\
    &+ \frac{1}{2}(\theta-\theta_{i})\nabla^{2}Q(\tilde{\theta}_{i})\cdot(\theta-\theta_{i})^{t}
\end{align*}
where $\nabla Q$ and $\nabla^{2}$ are gradient and Hessian matrix of $Q$ at $\tilde{\theta}_{i}$, respectively, and $Q=\log \left(f_{i}(y\vert \theta)g_{i}(\theta)\right)$.
Therefore,

\begin{align*}
    \P(y\vert H_{i})&=\int_{\Theta_{i}}{\rm exp}\left(Q(\tilde{\theta}_{i}) + (\theta-\theta_{i})\cdot\nabla Q(\tilde{\theta}_{i}) +  \frac{1}{2}(\theta-\theta_{i})\nabla^{2}Q(\tilde{\theta}_{i})\cdot(\theta-\theta_{i})^{t}\right)d\theta\\
    &=\int_{\Theta_{i}}{\rm exp}\left(Q(\tilde{\theta}_{i}) +  \frac{1}{2}(\theta-\theta_{i})\nabla^{2}Q(\tilde{\theta}_{i})\cdot(\theta-\theta_{i})^{t}\right)d\theta
\end{align*}
where in the last line we used that $\tilde{\theta}_{i}$ is stationary point for $Q$. Hence,

\begin{align*}
    \P(y\vert H_{i})&={\rm exp}(Q(\tilde{\theta}_{i}))\int_{\Theta_{i}}{\rm exp}\left(\frac{1}{2}x\nabla^{2}Q(\tilde{\theta}_{i})\cdot x^{t}\right)dx.
\end{align*}
Let us now introduce
\begin{align}\label{matrix}
    P_{i}=-\nabla^{2}Q(\tilde{\theta}_{i}).
\end{align}
We write now
\begin{align*}
     \P(y\vert H_{i})={\rm exp}(Q(\tilde{\theta}_{i}))\int_{\Theta_{i}}{\rm exp}\left(-\frac{1}{2}xP_{i}\cdot x^{t}\right)dx.
\end{align*}
Since $\tilde{\theta}_{i}$ is a point of maximum for $Q$, then $P_{i}$ is positive definite and we may apply a decomposition in singular values. That is, 
\begin{align*}
    P_{i}=U\Sigma_{i}U^{t}
\end{align*}
where $U$ is an orthogonal matrix. So that,
\begin{align*}
     \P(y\vert H_{i})={\rm exp}(Q(\tilde{\theta}_{i}))\int_{U\Theta_{i}}{\rm exp}\left(-\frac{1}{2}y\Sigma_{i}\cdot y^{t}\right)dy\\
     ={\rm exp}(Q(\tilde{\theta}_{i}))\int_{\Theta_{i}}{\rm exp}\left(-\frac{1}{2}y\Sigma_{i}\cdot y^{t}\right)dy.
\end{align*}
From here, we finally obtain

\begin{align*}
    \P(y\vert H_{i})&={\rm exp}(Q(\tilde{\theta}_{i}))\Pi_{j=1}^{p}\sqrt{\frac{2\pi}{\lambda_{j}}}\\
    &=f_{i}(y\vert \tilde{\theta}_{i})g_{i}(\tilde{\theta}_{i})\Pi_{j=1}^{p}\sqrt{\frac{2\pi}{\lambda_{j}}}\\
    &=f_{i}(y\vert \tilde{\theta}_{i})g_{i}(\tilde{\theta}_{i})\frac{(2\pi)^{p/2}}{\det(P_{i})^{1/2}}.
\end{align*}
where $(\lambda_{j})$ are the eigenvalues of $P_{i}$.
 \end{proof}
 
\begin{proof}[Proof of Theorem \ref{betterthm}]
 Formula \eqref{logbif} is a simply consequence of Corollary \ref{1cor} while the last formula \eqref{betterbic} requires only a few easy  computations. Indeed, for simplcity, let us write 
 \begin{align*}
     \log(L_{i})=\log(f_{i}(y\vert \tilde{\theta}_{i})g_{i}(\tilde{\theta}_{i})), \quad \text{for $i=0,1.$}
 \end{align*}
 Since $\vert \theta_{1}\vert=\vert \theta_{0}\vert$ we immediately have
\begin{align*}
     2\log(\B_{0,1})=&2 \log(L_{0})-2\log(L_{1})+\vert\theta_{1}\vert\log(m)- \vert\theta_{0}\vert\log(m) \\
     &+ \log(\det(P_{1}))-\log(\det(P_{0}))
\end{align*}
where $m$ is the  number of observations in $\Theta_{i}$ for $i=0,1$. From here, we have by definition of $\nbic$ that
\begin{align*}
    2\log(\B_{0,1}) &= \nbic(H_{1})-\nbic(H_{0})\\
    &=\nbic(\Delta_{1,0}).
\end{align*}
\end{proof}
Lastly, let us provide the proof of \autoref{thm:consistency}. 

\begin{proof}
Recall that a model selection criterion $IC$ is consistent if, given a true data-generating model $M_0$ (with $p_0$ parameters) and an overfitted competing model $M_1$ (with $p_1$ parameters, $p_1 > p_0$ such that $M_0 \subset M_1$), the probability of selecting the true model approaches one as the sample size increases:

\begin{align*}
\lim_{n \rightarrow \infty} \mathbb{P}(IC(M_0) < IC(M_1)) = 1
\end{align*}
To prove consistency, we must show that the difference $\Delta \text{BIC\_HES} = \text{BIC\_HES}(M_1) - \text{BIC\_HES}(M_0)$ is positive as $n \rightarrow \infty$. We consider the difference in the criterion values between the overfitted model $M_1$ and the true model $M_0$:

\begin{align*}
\Delta \text{BIC\_HES} = \text{BIC\_HES}(M_1) - \text{BIC\_HES}(M_0).
\end{align*}
Substituting the definition of $\text{BIC\_HES}$:
\begin{align*}
\Delta \text{BIC\_HES} &= \left[-2 \log L_1 + p_1 \log n + \log|J_1|\right] - \left[-2 \log L_0 + p_0 \log n + \log|J_0|\right] \\
&= \underbrace{-\left[2(\log L_1 - \log L_0)\right]}_{\text{Likelihood Ratio Term ($\Lambda$)}} + \underbrace{(p_1 - p_0) \log n}_{\text{Penalty Term}} + \underbrace{(\log |J_1| - \log |J_0|)}_{\text{Hessian Difference Term ($\Delta H$)}}
\end{align*}
We analyze the asymptotic order of each component as $n \rightarrow \infty$. The likelihood ratio statistic is defined as
\begin{align*}
\Lambda = 2(\log L_1 - \log L_0).
\end{align*}
Under the Laplace approximation, write the parameters of the larger model as 
\(\theta_1 = (\theta_0, \eta)\), where \(\eta \in \mathbb{R}^{p_1 - p_0}\) are the extra parameters. 
A second-order Taylor expansion of the log-likelihood around \((\theta_0, 0)\) gives

\begin{align*}
\log L_1(\hat{\theta}_1) 
&= \log L_1(\theta_0, 0) 
+ (\hat{\eta})^\top \frac{\partial \log L_1}{\partial \eta}\Big|_{(\theta_0,0)} 
+ \frac{1}{2} \hat{\eta}^\top J_{\eta\eta} \hat{\eta} 
+ o_p(\|\hat{\eta}\|_{2}^2),
\end{align*}
where \(J_{\eta\eta}\) is the Hessian block corresponding to the extra parameters \(\eta\) and the remainder term satisfies \(o_p(\|\hat{\eta}\|_{2}^2) \to 0\) in probability faster than \(\|\hat{\eta}\|_{2}^2\), and \(\hat{\eta}\) are the MLEs of the extra parameters. Since \(\theta_0\) maximizes the likelihood of the true model \(M_0\), the score with respect to the original parameters vanishes:
\begin{align*}
\frac{\partial \log L_1}{\partial \theta}\Big|_{(\theta_0,0)} = 0.
\end{align*}
Partitioning the Hessian and the parameter difference according to \(\theta_0\) and \(\eta\), the quadratic term reduces asymptotically to

\begin{align*}
\Lambda = 2(\log L_1 - \log L_0) = \hat{\eta}^\top J_{\eta\eta} \hat{\eta} + + o_p(\|\hat{\eta}\|_{2}^2),
\end{align*}
where \(J_{\eta\eta} = - \frac{\partial^2 \log L_1}{\partial \eta \, \partial \eta^\top}\big|_{(\theta_0,0)}\) is the Hessian block corresponding to the extra parameters. Notice that

\begin{align*}
\sqrt{n} \hat{\eta} \xrightarrow{d} \mathcal{N}(0, I_{\eta\eta}^{-1}),
\end{align*}
where

\begin{align*}
I_{\eta\eta} = \mathbb{E}\Big[-\frac{\partial^2 \log L_1}{\partial \eta \, \partial \eta^\top}\big|_{(\theta_0,0)}\Big]
\end{align*}
is the Fisher information matrix for the extra parameters. Therefore, the likelihood ratio statistic satisfies

\begin{align*}
\Lambda = 2(\log L_1 - \log L_0) 
= \hat{\eta}^\top J_{\eta\eta} \hat{\eta} + o_p(\|\hat{\eta}\|_{2}^2) 
\xrightarrow{d} \chi^2_{p_1 - p_0}.
\end{align*}
Now,  recall that a sequence of random variables \(X_n\) is called stochastically bounded, written \(X_n = O_p(1)\), if for every \(\varepsilon > 0\) there exists a constant \(M > 0\) such that

\begin{align*}
\sup_n \mathbb{P}(|X_n| > M) < \varepsilon.
\end{align*}
Since \(\Lambda\) converges in distribution to a chi-squared variable, which is  finite with probability 1, we can choose \(M\) large enough that \(\mathbb{P}(\Lambda > M) < \varepsilon\) for all sufficiently large \(n\). Therefore, by definition,
\begin{align*}
\Lambda = O_p(1).
\end{align*}
The complexity penalty term is $(p_1 - p_0) \log n$. Since $M_1$ is an overfitted model, we have $p_1 > p_0$. Then

\begin{align*}
\lim_{n \rightarrow \infty} (p_1 - p_0) \log n = \infty.
\end{align*}
This term is of the order $O(\log n)$.

The Hessian difference term is $\Delta H = \log |J_1| - \log |J_0|$. Based on the Laplace expansion of the marginal likelihood, the term $\log |J|$ captures parameter dependency and is $O_p(1)$ with respect to $n$. So

\begin{align*}
\Delta H = O_p(1).
\end{align*}
Substituting the asymptotic orders back into the expression for $\Delta \text{BIC\_HES}$:
\begin{align*}
\Delta \text{BIC\_HES} = \underbrace{- O_p(1)}_{\text{Likelihood Gain}} + \underbrace{O(\log n)}_{\text{Complexity Penalty}} + \underbrace{O_p(1)}_{\text{Hessian Difference}}
\end{align*}
As $n \rightarrow \infty$, the term $O(\log n)$ dominates the $O_p(1)$ terms.  Note that, 

\begin{align*}
\Delta \text{BIC\_HES} \approx (p_1 - p_0) \log n + O_p(1).
\end{align*}
Since $p_1 - p_0 > 0$, the dominant term ensures that $\Delta \text{BIC\_HES}$ tends to positive infinity.

Therefore, the probability that the criterion of the larger model $M_1$ is greater than the criterion of the true model $M_0$ approaches 1. That is,

\begin{align*}
\lim_{n \rightarrow \infty} \mathbb{P}(\Delta \text{BIC\_HES} > 0) = 1.
\end{align*}
\end{proof}

\section{Conclusions}\label{sec:futureworks}
In this paper, we proposed a perturbated Bayesian Information Criterion (BIC$\_$HES) that enhances traditional model selection by incorporating the log-determinant of the Hessian matrix of the log-likelihood function. This modification introduces a data-dependent complexity penalty that offers a more adaptive and geometrically informed approach to model comparison. Through extensive simulations across linear, mixed, binomial, and Poisson models, we demonstrated that BIC$\_$HES consistently outperforms classical criteria such as BIC, AIC, and AICc, particularly in small-sample scenarios and in the presence of noise variables. The criterion shows strong robustness and improved accuracy in identifying the true data-generating model. From a practical standpoint, BIC$\_$HES is compatible with standard statistical software and can be implemented using posterior draws or observed Fisher information matrices. Future work may explore extensions to high-dimensional settings, regularized Bayesian frameworks, and nonparametric models. Overall, BIC$\_$HES represents a promising advancement in Bayesian model selection, bridging theoretical rigor with practical applicability.

\section*{Acknowledgements}
AMH has been supported by project PRIN 2022 
``understanding the LEarning process of QUantum Neural networks (LeQun)'', proposal code 2022WHZ5XH -- CUP J53D23003890006. The author AMH is a member of the ``Gruppo Nazionale per l'Analisi Matematica, la Probabilità e le loro Applicazioni (GNAMPA)'' of the ``Istituto Nazionale di Alta Matematica ``Francesco Severi'' (INdAM)''.

\section*{Compliance with Ethical Standards}
\begin{itemize}
    \item Disclosure of potential conflicts of interest: The authors confirm that there are no conflicts of interest.
    \item Research involving Human Participants and/or Animals: The authors confirm that our research paper does not involve human participants and/or animals as a dataset.
\end{itemize}

\bibliographystyle{siam}
\bibliography{biblio}

@book{amari2016information,
  title={Information geometry and its applications},
  author={Amari, Shun-ichi},
  volume={194},
  year={2016},
  publisher={Springer}
}

@article{gelman2006prior,
  title={Prior distributions for variance parameters in hierarchical models (comment on article by Browne and Draper)},
  author={Gelman, Andrew},
  year={2006}
}

@article{perez2017scaled,
  title={The scaled beta2 distribution as a robust prior for scales},
  author={P{\'e}rez, Maria-Egl{\'e}e and Pericchi, Luis Ra{\'u}l and Ram{\'\i}rez, Isabel Cristina},
  year={2017}
}

@article{gelman2014understanding,
  title={Understanding predictive information criteria for Bayesian models},
  author={Gelman, Andrew and Hwang, Jessica and Vehtari, Aki},
  journal={Statistics and computing},
  volume={24},
  number={6},
  pages={997--1016},
  year={2014},
  publisher={Springer}
}

@article{bozdogan1987model,
  title={Model selection and Akaike's information criterion (AIC): The general theory and its analytical extensions},
  author={Bozdogan, Hamparsum},
  journal={Psychometrika},
  volume={52},
  number={3},
  pages={345--370},
  year={1987},
  publisher={Springer-Verlag}
}

@article{spiegelhalter2014deviance,
  title={The deviance information criterion: 12 years on},
  author={Spiegelhalter, David J and Best, Nicola G and Carlin, Bradley P and Linde, Angelika},
  journal={Journal of the Royal Statistical Society Series B: Statistical Methodology},
  volume={76},
  number={3},
  pages={485--493},
  year={2014},
  publisher={Oxford University Press}
}

@article{held2014applied,
  title={Applied statistical inference},
  author={Held, Leonhard and Bov{\'e}, D Saban{\'e}s},
  journal={Springer, Berlin Heidelberg, doi},
  volume={10},
  number={978-3},
  pages={16},
  year={2014},
  publisher={Springer}
}

@article{faulk2018,
  title={Computing Bayes factors to measure evidence from experiments: An extension of the BIC approximation},
  author={Faulkenberry, Thomas J},
  journal={arXiv preprint arXiv:1803.00360},
  year={2018}
}

@article{raftery1995,
  title={Bayesian model selection in social research},
  author={Raftery, Adrian E},
  journal={Sociological methodology},
  pages={111--163},
  year={1995},
  publisher={JSTOR}
}

@article{vehtari2002,
  title={Bayesian model assessment and comparison using cross-validation predictive densities},
  author={Vehtari, Aki and Lampinen, Jouko},
  journal={Neural computation},
  volume={14},
  number={10},
  pages={2439--2468},
  year={2002},
  publisher={MIT Press}
}

@incollection{akaike1998,
  title={Information theory and an extension of the maximum likelihood principle},
  author={Akaike, Hirotogu},
  booktitle={Selected papers of hirotugu akaike},
  pages={199--213},
  year={1998},
  publisher={Springer}
}

@article{spiegelhalter1996,
  title={BUGS 0.5: Bayesian inference using Gibbs sampling manual (version ii)},
  author={Spiegelhalter, David and Thomas, Andrew and Best, Nicky and Gilks, Wally},
  journal={MRC Biostatistics Unit, Institute of Public Health, Cambridge, UK},
  pages={1--59},
  year={1996}
}

@article{watanabe2010,
  title={Asymptotic equivalence of Bayes cross validation and widely applicable information criterion in singular learning theory.},
  author={Watanabe, Sumio and Opper, Manfred},
  journal={Journal of machine learning research},
  volume={11},
  number={12},
  year={2010}
}

@article{van2021bayes,
  title={Bayes factors for mixed models},
  author={van Doorn, Johnny and Aust, Frederik and Haaf, Julia M and Stefan, Angelika M and Wagenmakers, Eric-Jan},
  journal={Computational Brain \& Behavior},
  pages={1--13},
  year={2021},
  publisher={Springer}
}

@article{vehtari2017,
  title={Practical Bayesian model evaluation using leave-one-out cross-validation and WAIC},
  author={Vehtari, Aki and Gelman, Andrew and Gabry, Jonah},
  journal={Statistics and computing},
  volume={27},
  number={5},
  pages={1413--1432},
  year={2017},
  publisher={Springer}
}

@article{gelman2021,
  title={What are the most important statistical ideas of the past 50 years?},
  author={Gelman, Andrew and Vehtari, Aki},
  journal={Journal of the American Statistical Association},
  volume={116},
  number={536},
  pages={2087--2097},
  year={2021},
  publisher={Taylor \& Francis}
}

@article{l1990,
  title={Nonlinear mixed effects models for repeated measures data},
  author={Lindstrom, Mary J and Bates, Douglas M},
  journal={Biometrics},
  pages={673--687},
  year={1990},
  publisher={JSTOR}
}

@article{rouder2012,
  title={Default Bayes factors for ANOVA designs},
  author={Rouder, Jeffrey N and Morey, Richard D and Speckman, Paul L and Province, Jordan M},
  journal={Journal of mathematical psychology},
  volume={56},
  number={5},
  pages={356--374},
  year={2012},
  publisher={Elsevier}
}

@article{gabry2020,
  title={Bayesian applied regression modeling via Stan},
  author={Gabry, J and Goodrich, B},
  journal={Package “rstanarm},
  year={2020}
}

@article{piironen2017,
  title={Comparison of Bayesian predictive methods for model selection},
  author={Piironen, Juho and Vehtari, Aki},
  journal={Statistics and Computing},
  volume={27},
  number={3},
  pages={711--735},
  year={2017},
  publisher={Springer}
}

@article{wagenmakers2007,
  title={A practical solution to the pervasive problems ofp values},
  author={Wagenmakers, Eric-Jan},
  journal={Psychonomic bulletin \& review},
  volume={14},
  number={5},
  pages={779--804},
  year={2007},
  publisher={Springer}
}

@book{pawitan2001,
  title={In all likelihood: statistical modelling and inference using likelihood},
  author={Pawitan, Yudi},
  year={2001},
  publisher={Oxford University Press}
}

@article{wallas,
  title={Study of Shannon entropy in the context of quantum mechanics: An application to free and confined harmonic oscillator},
  author={Wallas S. Nascimento, Frederico V. Prudente},
  journal={Química Nova},
  volume={39},
  number={6},
  pages={757--764},
  year={2016},
  publisher={Scielo Brasil}
}

@Manual{rsoftware,
    title = {R: A Language and Environment for Statistical Computing},
    author = {{R Core Team}},
    organization = {R Foundation for Statistical Computing},
    address = {Vienna, Austria},
    year = {2022},
    url = {https://www.R-project.org/},
    version = {4.3.2},
  }

@article{neath1997regression,
 title={Regression and time series model selection using variants of the Schwarz information criterion},
author={Neath, Andrew A and Cavanaugh, Joseph E},
journal={Communications in Statistics - Theory and Methods},
volume={26},
number={3},
pages={559--580},
year={1997},
publisher={Taylor \& Francis},
doi={10.1080/03610929708831934}
}

@Manual{rizopoulos2023,
  title = {GLMMadaptive: Generalized Linear Mixed Models using Adaptive Gaussian Quadrature},
  author = {Dimitris Rizopoulos},
  year = {2023},
  note = {https://drizopoulos.github.io/GLMMadaptive/, https://github.com/drizopoulos/GLMMadaptive},
}

@article{debruine2021,
author = {Lisa M. DeBruine and Dale J. Barr},
title ={Understanding Mixed-Effects Models Through Data Simulation},
journal = {Advances in Methods and Practices in Psychological Science},
volume = {4},
number = {1},
pages = {2515245920965119},
year = {2021},
doi = {10.1177/2515245920965119},
URL = { https://doi.org/10.1177/2515245920965119},
eprint = {https://doi.org/10.1177/2515245920965119}
}

@incollection{bozdogan1988icompcriterion,
  title={ICOMP: A New Model-Selection Criterion},
  author={Bozdogan, Hamparsum},
  booktitle={Classification and Related Methods of Data Analysis},
  editor={Bock, Hans H.},
  year={1988},
  publisher={Elsevier Science Publishers B. V. (North-Holland)},
  address={Amsterdam},
  pages={599--608}
}

@article{HURVICHTSAI1989,
    author = {Hurvich, Clifford M. and Tsai, Chih-Ling},
    title = {Regression and time series model selection in small samples},
    journal = {Biometrika},
    volume = {76},
    number = {2},
    pages = {297-307},
    year = {1989},
    month = {06},
    issn = {0006-3444},
    doi = {10.1093/biomet/76.2.297},
    url = {https://doi.org/10.1093/biomet/76.2.297},
    eprint = {https://academic.oup.com/biomet/article-pdf/76/2/297/737009/76-2-297.pdf},
}

@inproceedings{akaike1973,
  author = {Akaike, Hirotugu},
  title = {Information Theory and an Extension of the Maximum Likelihood Principle},
  booktitle = {Proceedings of the 2nd International Symposium on Information Theory},
  editor = {Petrov, Boris N. and Csaki, F.},
  year = {1973},
  pages = {267-281},
  publisher = {Akademiai Kiado},
  address = {Budapest}
}
 \end{document}